\documentclass{article}
\usepackage[a4paper, margin=1in]{geometry}
\usepackage{microtype}
\usepackage{graphicx}
\usepackage{subcaption}

\usepackage{hyperref}

\usepackage{amsmath}
\usepackage{amssymb}
\usepackage{mathtools}
\usepackage{amsthm}

\theoremstyle{plain}
\newtheorem{theorem}{Theorem}[section]

\newtheorem{lemma}[theorem]{Lemma}

\theoremstyle{definition}
\newtheorem{definition}[theorem]{Definition}

\theoremstyle{remark}

\newtheorem{example}[theorem]{Example}
\newtheorem{claim}[theorem]{Claim}

\newcommand{\abs}[1]{\left|#1\right|}

\DeclareMathOperator*{\E}{\mathbb{E}}
\newcommand{\cD}{\mathcal{D}}
\newcommand{\cX}{\mathcal{X}}

\newcommand{\cA}{\mathcal{A}}
\newcommand{\cF}{\mathcal{F}}

\newcommand{\cS}{\mathcal{S}}

\newcommand{\one}{\mathbf 1}

\usepackage[textsize=tiny]{todonotes}

\title{Fairness in Limited Resources Settings}

\begin{document}
\author{
Eitan Bachmat\\
Ben-Gurion University of the Negev \thanks{ebachmat@bgu.ac.il}
\and
Inbal Livni Navon\\
Ben-Gurion University of the Negev \thanks{inballn@bgu.ac.il}
}
\date{}
\maketitle

\begin{abstract} 
In recent years many important societal decisions are made by machine-learning algorithms, and many such important decisions have strict capacity limits, allowing resources to be allocated only to the highest utility individuals. For example, allocating physician appointments to the patients most likely to have some medical condition, or choosing which children will attend a special program. When performing such decisions, we consider both the prediction aspect of the decision and the resource allocation aspect.
In this work we focus on the fairness of the decisions in such settings. The fairness aspect here is critical as the resources are limited, and allocating the resources to one individual leaves less resources for others. When the decision involves prediction together with the resource allocation, there is a risk that information gaps between different populations will lead to a very unbalanced allocation of resources.

We address settings by adapting definitions from resource allocation schemes, identifying connections between the algorithmic fairness definitions and resource allocation ones, and examining the trade-offs between fairness and utility.
 We analyze the price of enforcing the different fairness definitions compared to a strictly utility-based optimization of the predictor, and show that it can be unbounded. 
 We introduce an adaptation of proportional fairness and show that it has a bounded price of fairness, indicating greater robustness, and propose a variant of equal opportunity that also has a bounded price of fairness.

\end{abstract}
\section{Introduction}
In recent years machine-learning algorithms have been used in many applications across various settings.  
Many important use cases of machine-learning algorithms arise in limited-resources settings, where our goal is to allocate a resource to the individuals that will benefit from it the most. 
In the medical domain, specialist appointments or hospital beds  are limited and allocated to the highest-risk patients. In education, special programs for excellent students often have strict capacity, and the slots are given to the highest potential students. Institutes that allocate funding, such as grant money or start-up investments, have a fixed budget, and their goal is to allocate the funds to the most promising projects. 

In all of these settings, the machine-learning decision involves two different considerations - risk prediction and resource allocation. The two considerations exist both when the machine-learning algorithm returns the allocation itself and when the algorithm is a risk predictor that we threshold to obtain an allocation. In the latter case, the allocation and risk prediction are clearly separated, but they both co-exist implicitly when the machine-learning algorithm returns the allocation itself. When discussing the quality of the prediction, we aim to discuss both options, where in the case of direct allocation function, the prediction is implicit. 

When the resource is very limited and the risk prediction has large uncertainty, gaps in the quality of the prediction across different populations can cause a large discrepancy in the allocation of resources across these groups. In the absence of fairness considerations, if we can better recognize the highest-risk individuals from a group $S_1$, the algorithm may allocate almost all of the resources to individuals that belong to $S_1$, see Figure \ref{fig:intro-single-threshold}. This can occur, for example, as a result of information gaps between different groups.  Such situations are often unwanted by the decision makers. If a city develops a special program for gifted children, we do not want all of them to come from the same school or neighborhood. 

Information gaps between individuals belonging to different groups appear in multiple important settings. In the medical setting, they can result from differences in the quality of medical history between individuals from different demographic groups \cite{zink2024race}. In education, they may result from the structure of the school system or students' behavior \cite{bird2025algorithms}. 
In sports, we have more information on the potential of older children compared to younger ones \cite{barth2024quantifying}. When allocating grants in research, newly appointed faculty members have larger uncertainty compared to older and more established ones. When decisions are made by a professional authority, they sometimes compensate for the lack of information. If we are to consider algorithmic decisions, we should formalize our goals in compensating information gaps.

\begin{figure}[t]
    \centering
    \includegraphics[width=\columnwidth]{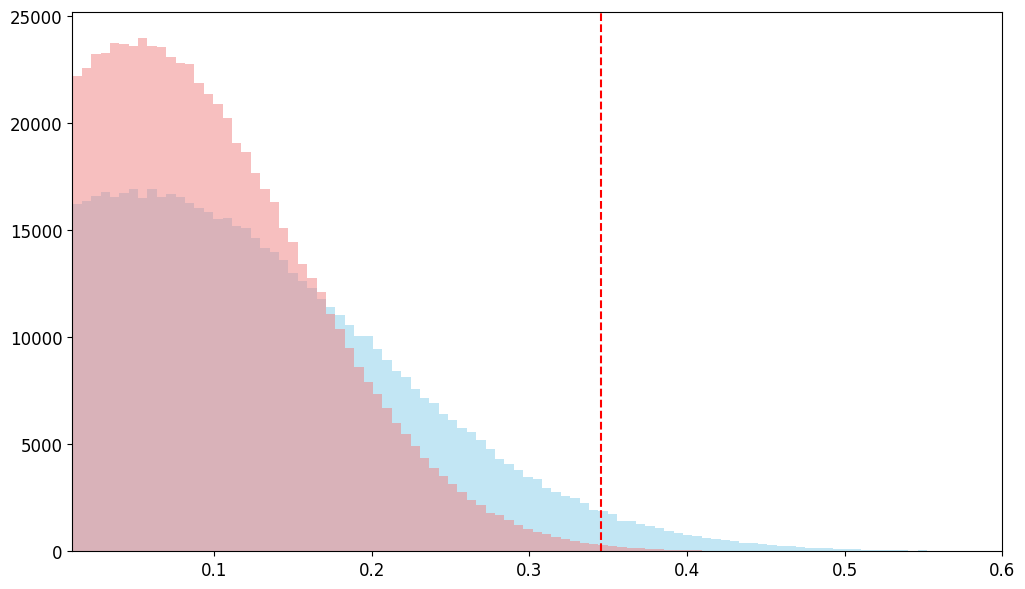}
    \caption{Simulating the output distribution of a predictor $p$ with a larger variance on $S_1$ (in blue) compared to $S_2$. The red line is a capacity limit of $1\%$ of the population, and it contains a large majority of individuals from $S_1$. See Appendix \ref{app:simulations} for more information.
    }
    \label{fig:intro-single-threshold}
\end{figure}

In the algorithmic fairness literature, there is a large body of work discussing bias in algorithms and proposing different methods for mitigating it. One specific approach that is aimed at compensating for information gaps between populations is the equal opportunity fairness requirement \cite{hardt2016equality}. This requirement was originally developed for the general setting, without a strict capacity. In our work, we show that it has limitations when applied to our setting of strict capacity constraints.

\subsection{Our contributions}
Our goal in this work is to develop a theoretical understanding of group fairness in the allocation of limited resources based on individual risk. We adapt resource-allocation fairness definitions to our setting of allocation based on prediction, and find connections between these definitions and algorithmic fairness definitions. We bound the price of fairness of different definitions, and suggest a new definition that admits a bounded price of fairness.

\paragraph{Adapting resource allocation definitions}
In Section \ref{sec:adapt} we adapt resource allocation fairness definitions to the setting of allocating resources to individuals based on their risk. In the resource allocation setting, we have a limited resource that we want to divide between different agents $i\in[m]$, each with their own utility function $u_i$ \cite{bertsimas2011price,bin2018fairness}. We adapt the definitions of max-min fairness and proportional fairness by identifying each protected demographic group with an agent, and defining a utility function for each such group. For each protected group $S$ that is identified with an agent $i$, we define $u_i$, i.e. the utility of an allocation $f$ is $u_i=\E[f(x)y|x\in S]$. For the proportional fairness definition, we also change the normalization to take into account that $S$ is a group of individuals and not a single agent.

\paragraph{Connecting group fairness definitions to resource allocation fairness} In Section \ref{sec:equiv} we show under which conditions equal opportunity and max-min fairness are identical. These conditions include a fully utilized setting, in which we use all of the available resources and each group can benefit from additional resources. In addition, we must have equal base rates, and our class of potential allocation functions has to be rich enough.

In the original definition of proportional fairness, the goal is to maximize the \emph{product} of all the utilities across agents $\Pi_i u_i$. In our adapted prediction setting, we show that requiring proportional fairness is a type of multiplicative regularization. Satisfying proportional fairness is equivalent to maximizing an expression with two terms, one of which is the distance to max-min fairness, expressed by the Kullback–Leibler divergence. The second is the log of the true positive count, and maximizing it maximizes the utility.

\paragraph{Finding the price of fairness} 
In the limited capacity scenario, it is important that the existing resources are distributed ethically, but we also care  about the best usage of the limited resources. In Section \ref{sec:price-of-fairness} we show that for many natural distributions with information gap, equal opportunity (with equal base rates) and max-min fairness allocate most resources to the group for which we have less information. That is, they are stricter than demographic parity, unlike in the standard setting considered in the equal opportunity paper \cite{hardt2016equality}. In some extreme cases, requiring equal opportunity or max-min fairness implies allocating nearly all resources to a group for which we have almost no information, causing a very large price of fairness. 

The utility regularization term in proportional fairness implies that we are less strict than max-min fairness, i.e. in the trade-off between fairness and maximizing the true positive count, leading to a higher true positive count. While there are some settings for which requiring proportional fairness implies most resources are given to the group we have less information on, we show that it cannot lead to diminishing true positive count. That is, the price of fairness is bounded for our adaptation of proportional fairness.

\paragraph{Suggesting new fairness definitions}
In the previous part, we noted that the price of fairness in limited-resources settings with information gaps can be high. In Section \ref{sec:new-def}, we suggest a generalization of equal opportunity that is less strict, such that the price of fairness is bounded.
Equal opportunity requires the same allocation to the group of individuals with $y=1$ in each group: equalizing $\E[f(x)|y=1,x\in S]$. In our new definition, we suggest instead considering only the \emph{achievable} true positives. That is, we consider only individuals with $y=1$ we could have obtained if we allocated all of our existing resources to this group. The idea behind our new definition is that all the other individuals with $y=1$ are out of our reach in any case, so they should not affect our fairness requirement. 
We can also generalize proportional fairness in the same way.

\subsection{Related works}
The problem of fair allocations of resources is a rich research area in many different settings, such as networks \cite{ogryczak2014fair}, facility locations \cite{marsh1994equity}, division of goods \cite{amanatidis2022fair}, operation research \cite{bertsimas2011price} and more. These very different settings still share many of the fairness definitions used, such as max-min fairness.  

In this work, we focused on fairness in resource allocation, where the allocation is based on individual risk prediction. There are works focusing on fairness in allocation combining other types of learning settings, such as predicting demand \cite{donahue2020fairness}, learning gradually revealed utilization \cite{elzayn2019fair}, federated learning \cite{zhangproportional}, networks \cite{wang2018machine} and additional areas.

In fairness in machine-learning, the goal is to ensure that the algorithm is not biased against any group or individual. Hardt et. al. \cite{hardt2016equality} suggest the definition of equal opportunity to compensate for information gaps between groups. For risk prediction, \cite{hebert2018multicalibration}, suggest extending calibration to many overlapping groups.  Gopalan et. al. \cite{gopalan2021omnipredictors} showed that a multicalibrated predictor can be post-processed to minimize different losses in an optimal way, and Hu et. al. \cite{hu2023omnipredictors} showed that it can be done under capacity and fairness constraints. Based on this framework, in parts of this work we model the allocation process as finding a multicalibrated predictor and post-processing it. 

There are different works focusing on the price of enforcing  different fairness definitions. In resource allocation \cite{bertsimas2011price}, it is usually modeled by comparing the utility with and without fairness constraints. In machine-learning \cite{menon2018cost,corbett2017algorithmic}, it can be modeled both by the harm to the accuracy, or by defining a utility function, as we do in this work. 

When we use resource allocation fairness definitions such as max-min and proportional fairness, we model the utility as the true positive count. For a given capacity and base rates, maximizing the true positive count is equal to minimizing the loss, although our definitions of max-min and proportional fairness are not equivalent to loss equality. Modeling the utility as the loss is done in the literature, for example in \cite{donini2018empirical,diana2021minimax,zhangproportional}.

\subsection{Open questions and discussion}
Our work opens many new research directions, both theoretical and applied. In this work, we suggest new fairness definitions, both adapted from resource allocations (proportional fairness) and new (achievable equal opportunity). We theoretically bound the price of fairness of existing and new definitions under minimal assumptions on the predictor distributions. An interesting follow-up work is to analyze the price of fairness for specific types of predictor distributions, which can be done both theoretically and experimentally. In our work, we showed that proportional fairness is in fact a type of multiplicative regularization. While multiplicative regularization terms exists in some settings, they are not often used, and it is a very interesting open question to compare the standard additive regularization to the multiplicative one.

\section{Preliminaries and notations}
We focus on the setting of prediction for resource allocation, with limited resources. We assume that there is a distribution over the population $\cD\subseteq\cX\times \{0,1\}$, where $x\in\cX$ is the features of an individual and $y\in\{0,1\}$ is an unknown binary output. The resource allocation goal is to allocate the resource to the individuals most likely to have $y=1$. We assume that our set of individuals $\cX$ is divided into $m$ disjoint groups $S_1,\ldots,S_m$. We denote by $g:\cX\rightarrow[m]$ the function that returns for every $x$ the group it belongs to.

In this work we focus on a limited-resources setting, where we want to find the optimal set of individuals to allocate the resources to, taking into account group fairness consideration. We allow for randomized allocations. A function $f:\cX\rightarrow[0,1]$ is a randomized allocation, where we treat values $f(x)=r\in(0,1)$ as allocating the resource to $x$ with probability $r$. We formally define the capacity and capacity within a group for such randomized allocations.
\begin{definition}
    A function $f:\cX\rightarrow [0,1]$ has capacity $c$ under distribution $\cD$ if $\E_{(x,y)\sim\cD}[f(x)]=c$. For a group $S_i\subset\cX$ the function $f$ has capacity $c$ on $S_i$ if $\E_{(x,y)\sim\cD}[f(x)|x\in S_i]=c$.
\end{definition}

In order to discuss group fairness in the setting of information gaps, we must first formally define it. We do so by breaking the process of finding the optimal allocation into a two-step process, where in the first we learn a predictor $p:\cX\rightarrow[0,1]$ with the aim of predicting $y$, and in the second we decide on the optimal set of individuals by post-processing $p$. From \cite{hu2023omnipredictors}, a predictor that is multicalibrated on a set of functions $\cF$ can be post-processed to satisfy group fairness constraints with loss as small as the best of $\cF$ (up to the calibration error). Therefore, breaking the decision process into two steps can be done in a way that does not harm the performance.

In this formalization of a two-step decision process, the information gap is modeled as the \emph{tail dominance} of $p$ on $S_1$ comparing to $p$ on $S_2$. The tail dominance mean that we can recognize more high-risk individuals in group $S_1$ comparing to $S_2$.
\begin{definition}\label{def:dominance}
    Let $p:\cX\rightarrow[0,1]$ be a predictor, and let $S_1,S_2\subset\cX$ be two disjoint groups.
    Group $S_1$ has $t_0$-tail dominance over group $S_2$ under distribution $\cD$ and predictor $p$ if for all $r\geq t_0$,
        \begin{align*}
        \Pr_{(x,y)\sim\cD}[p(x)> r|x\in S_1]\geq \Pr_{(x,y)\sim\cD}[p(x)> r|x\in S_2].
    \end{align*} 
\end{definition}

Finding the optimal allocation by post-processing a given predictor is formalized by applying a randomized threshold function to the predictor $p$. 
\begin{definition}
    For a set of $m$ groups, a randomized group threshold function $\tau$ is a function $\tau:[0,1]\times[m]\rightarrow[0,1]$, such that for each $i\in[m]$ there exist a threshold $t_i\in[0,1]$ and a value $\gamma_i\in[0,1]$ such that
    \begin{align*}
        \tau(r,i) = \begin{cases}
            1 \quad &r > t_i\\
            \gamma_i \quad &r=t_i\\
            0 \quad &r<t_i.
        \end{cases}
    \end{align*}
\end{definition}
For a predictor $p:\cX\rightarrow[0,1]$ and threshold function $\tau$, the allocation is defined by  $\forall x, f(x)=\tau(p(x),g(x))$. We note that we can similarly apply other post-processing functions $\tau$ that are not randomized threshold function.

As we are in a limited resources settings, we want to fully use all of our capacity. We show next that it is always possible when using randomized threshold functions.
\begin{claim}\label{claim:existence-threshold}
Let $\cD$ be a distribution over $\cX\times\{0,1\}$ with $\cX$ is divided into $m$ groups $S_1,\ldots,S_m$. Then, for every capacity $c\in[0,1]$ and group $i$ such that $\Pr_{x\sim\cD}[x\in S_i]>0$, there is a randomized group threshold function $\tau$ satisfying
$\E_{x\sim\cD}[\tau(p(x),i)|x\in S_i]=c$.
\end{claim}
The proof appears on Appendix \ref{app:def-proofs}. 
Thresholding a predictor is a specific type of post-processing of the predictor. In general, we can apply \emph{any} function $\tau:[0,1]\times[m]\rightarrow[0,1]$ to $p$. We focus on randomizes threshold functions as they take the highest risk individual from each grop.

Some of the proofs in the paper use quantile functions. A preliminaries section about quantile functions is on Appendix \ref{app:quant}. In appendix \ref{app:simulations} we have information about the graphs and simulations appearing in the paper.

\subsection{Loss minimization}
In most decision-making machine-learning setting, we optimize by minimizing the expected loss of a loss function. When the decision and outcome $y$ are both binary, the loss function is $\ell:\{0,1\}\times\{0,1\}\rightarrow[0,1]$. We define next the expected loss for a randomized assignment $f$. When $f$ is binary this is the same as $\E_{(x,y)\sim\cD}[\ell(f(x),y)]$.
\begin{definition}
    The expected loss of a function $f:\cX\rightarrow[0,1]$ under distribution $\cD$ and a binary loss function $\ell:\{0,1\}\times\{0,1\}\rightarrow[0,1]$ is $\E_{(x,y)\sim\cD}[f(x)\ell(1,y) + (1-f(x))\ell(0,y)]$.
\end{definition}

When we have full utilization, minimizing all standard binary loss functions $\ell:\{0,1\}\times\{0,1\}\rightarrow[0,1]$ are equivalent to maximizing the true positive count, defined next. We remark that this is not the true positive rate, as it is not normalized.
\begin{definition}
    The true positive count of a function $f:\cX\rightarrow[0,1]$ under a distribution $\cD$ is $\E_{(x,y)\sim\cD}[f(x)y]$.
\end{definition}

\begin{claim}\label{claim:loss-equiv}
Let $\cD$ be a distribution with $\Pr_{(x,y)\sim\cD}[y=1]\in(0,1)$. 
Then for every capacity $c\in[0,1]$, any family $\cF$ of functions $h:\cX\rightarrow[0,1]$ and binary loss $\ell$ with $\ell(0,0)=\ell(1,1)=0$, the function $f\in\cF$ minimizing the expected loss of $\ell$ and having capacity $c$ is the same as the function $f\in\cF$ maximizing the true positive count and having capacity $c$.
\end{claim}
The proof appears in Appendix \ref{app:def-proofs}. In the rest of the paper, we model the utility as maximizing the true positive count. We chose the true positive count since it represent the total number of successful allocation of the resource.

\subsection{Group fairness definitions}
We define below some of the the basic group fairness definitions. We state the definitions below without errors for simplicity, although when enforcing the definitions there might be a small error. Our results also hold for the cases of small errors in enforcing the fairness definition (see Theorem \ref{thm:eo-diff-err} for a precise statement).

\begin{definition}[Demographic parity]
    A function $f:\cX\rightarrow[0,1]$ satisfies demographic parity with respect to a distribution $\cD$ and groups $S_1,\ldots,S_m$ if for all $i,j\in[m]$, \[\E_{(x,y)\sim\cD}[f(x)|x\in S_i]=\E_{(x,y)\sim\cD}[f(x)|x\in S_j].\]
\end{definition}

\begin{definition}[Equal opportunity \cite{hardt2016equality}]
    A function $f:\cX\rightarrow[0,1]$ satisfies equal opportunity with respect to a distribution $\cD$ and groups $S_1,\ldots,S_m$ if for all $i,j\in[m]$, \[
        \E_{(x,y)\sim\cD}[f(x)|x\in S_i,y=1]=\E_{(x,y)\sim\cD}[f(x)|x\in S_j,y=1].\]
\end{definition}

For a predictor $p$, one fairness requirement is calibration within groups. We emphasize that this definition is only applicable to a predictor $p$, where $p(x)$ represent the probability of $y=1$. For an allocating function $f$, $f(x)$ implies that we allocate the resource to $x$ with probability $r$. 
\begin{definition}[Group calibration]
    A predictor $p:\cX\rightarrow \mathbb{R}$ is $\alpha$-calibrated on groups $S_1,\ldots,S_m$ with respect to a distribution $\cD$ if for all $i\in [m]$,
        \begin{align*}
        \int_0^1 \left|\E_{(x,y)\sim\cD}[(y-r)\wedge \one(p(x)=r)|x\in S_i]\right|dr \leq\alpha
    \end{align*}
\end{definition}
It is also possible to require calibration with respect to overlapping groups \cite{hebert2018multicalibration} or to a set of functions.

When post-processing a predictor to satisfy a fairness requirement there are two options. In the first, we assume that the predictor correctly describes the distribution of $y$ and enforce the definition based on this assumption, i.e. on the simulated distribution defined next. If we know the distribution of $p(x)$ on the groups, then enforcing a group fairness requirement under this assumption does not require any new samples. The approach is to use new samples and enforce the definition by approximating the fairness requirement on these samples. If $p$ is not well-calibrated, then the first method can incur a large error. In this work we mostly assume that $p$ is calibrated, and therefore use the first approach using a simulated distribution.
\begin{definition}\label{def:simulated}
    Let $\cD$ be a distribution over pairs $x\in\cX,y\in\{0,1\}$ and $p:\cX\rightarrow[0,1]$ be a predictor. Then the simulated distribution $\cD_p$ over pairs $x,y$ is the distribution of picking $x\sim\cX$ and then $y\sim\textbf{Ber}(p(x))$.
\end{definition}

In this work we quantify the price of enforcing different fairness definitions. We discuss both additive and multiplicative price of fairness.
\begin{definition}
    Let $\cD$ be a distribution over pairs $\cX\times\{0,1\}$ and let $\cF$ be a class of functions. Given any fairness requirement, let $\cF'\subset\cF$ be the functions if $\cF$ satisfying a certain fairness definition. Then the multiplicative price of fairness is given by
    \[\frac{\max_{h\in\cF}\E_{(x,y)\sim\cD}[h(x)y]}{\max_{f\in\cF'}\E_{(x,y)\sim\cD}[f(x)y]}.\]
    
    The additive price of fairness is defined by the expression \[\max_{h\in\cF}\E_{(x,y)\sim\cD}[h(x)y]-\max_{f\in\cF'}\E_{(x,y)\sim\cD}[f(x)y].\]
\end{definition}

\section{Resource allocation fairness definitions in the prediction setting}\label{sec:adapt}
In this section we adapt classical resource allocation fairness definitions to our setting of prediction. Resource allocations fairness definitions assume that there are different agents $i\in[m]$, and each has a \emph{utility function} $u_i$ receiving as an input the resources and outputting the utility. In our case, we adapt the definitions by modeling each protected group $S_i\subset \cX$ as an agent, and defining $u_i$ as the expected true positive count on the group $u_i=\E_{(x,y)\sim\cD}[f(x) y|x\in S_i]$.  

We note that modeling the utility as the expected true positive count is a design choice we have made in this work, and different choices yield different fairness definitions. While  Claim \ref{claim:loss-equiv} shows that minimizing any standard binary loss is equivalent to maximizing the expected true positive count, this equivalence doesn't hold when we compare the expected loss of two groups with different base rates or different capacities. Therefore, our definitions are not equivalent to modeling the utility as loss in the fairness requirements.
We choose the expected true positive count as our utility because it better represents the gain we have in most limited-resources applications. That is, we care about how many patients we successfully diagnose, or how many talented children are admitted to the special program. 

The resource allocation fairness definitions are \emph{optimization definitions}. Rather than defining minimal requirements for fairness, satisfying the resource allocation fairness requirement is maximizing some functions of the utilities. As we are in the prediction setting, this mean that in order to satisfy our resource allocation we must pick the optimal allocation $f$ from a class of functions $\cF$, in a similar manner to loss minimization. As such, the resource allocation fairness definitions could be viewed as both satisfying a fairness requirement and minimizing a loss.

We state below the adaptation of max-min fairness, where we replace the utility with the true positive count of each group. 
\begin{definition}
    For a family of functions $\cF$, a function $f:\cX\rightarrow[0,1]$ satisfies max-min fairness with respect to a distribution $\cD$ and groups $\cS = \{S_1,\ldots,S_m\}$ if
    \begin{align*}
        f = \arg\max_{h\in\cF}\min_{i\in[m]}\left\{\E_{(x,y)\sim\cD}[h(x) y | x\in S_i]\right\}
    \end{align*}
\end{definition}
We also use an adaptation of proportional fairness. 
The standard proportional fairness definition handles a setting of different agents and their utility. In our setting, the agents are groups of individuals, which results in a normalization issue with the sizes of the groups. When we discuss fairness with respect to groups, if we take a group $S_i$ and split it into two random subsets, we do not want the fairness requirement to change. Therefore, we consider our proportional fairness definition as \emph{normalized proportional fairness}, adding a normalization factor $\Pr_{(x,y)\sim\cD}[x\in S_i]$ for each group $S_i$. We prove in Appendix \ref{app:def-proofs} that our normalized proportional fairness is invariant under group splitting.
\begin{definition}\label{def:norm-prop}
    For a family of functions $\cF$, a function $f:\cX\rightarrow[0,1]$ satisfy normalized proportional fairness with respect to a distribution $\cD$ and groups $\cS = \{S_1,\ldots,S_m\}$ if
    \begin{align*}
        f =&\arg\max_{h\in\cF}\mathbf{Prop}_{\cD,\cS}(h)\\=&\arg\max_{h\in\cF}\sum_{i\in[m]}\Pr_{(x,y)\sim\cD}[x\in S_i]\log\E_{(x,y)\sim\cD}[h(x) y | x\in S_i].
    \end{align*}
\end{definition}

\section{Connections between fairness definitions}\label{sec:equiv}
In this part we show the connections between the group fairness definitions from algorithmic fairness and the adapted resource allocation fairness definitions. 

\subsection{Equal opportunity and max-min fairness}
We characterize when equal opportunity and max-min fairness are the same. We start by characterizing sufficient conditions for max-min fairness to imply equal true positive count among all groups (i.e. equal utility in the original definition). There are specific settings for which it is known that max-min fairness implies equal utility, such as in \cite{donini2018empirical}
and in this work we introduce another such setting.

The first condition is that we are in a non-degenerate setting, for which any additional allocation to a group improves its expected true positive count. This is common when the resources are limited, as we cannot recognize all of the $y=1$ in our population with capacity $c$. We remark that non-degenerate condition also means that $c$ is smaller than any of the groups.
\begin{definition}\label{def:non-deg}
    A distribution $\cD\subset \cX\times\{0,1\}$ is non degenerate with respect to groups $S_1,\ldots,S_m$, capacity $c$ and class of functions $\cF$ if for every $f\in\cF$ with $\Pr_{(x,y)\sim\cD}[f(x)=1]\geq 1-c$ and every $i\in[m]$, we have $\E_{(x,y)\sim\cD}[y|f(x)=1, x\in S_i] > 0$.
\end{definition}

The second condition is that the class of functions in rich enough. A sufficient condition is that $\cF$ is closed under group post-processing operations. We remark that the set of post-processing of a predictor is closed under post-processing.
\begin{definition}\label{def:closed-post}
    A family of functions $\cF$ is closed under group post-processing for groups $S_1,\ldots,S_m$ if for any $f\in\cF$ and any function $\tau:[0,1]\times[m]\rightarrow[0,1]$, the function $h$ defined by $h(x)=\tau(f(x),g(x))$ is in $\cF$.
\end{definition}

\begin{claim}\label{claim:max-min-equal}
    For $c<0.5$, let $\cD$ be a non degenerate distribution with respect to groups $S_1,\ldots,S_m$, capacity $2c$ and function family $\cF$ that is closed under group post-processing. Then the function $f\in\cF$ satisfying max-min fairness with respect to the groups $S_1,\ldots S_m$ and having capacity $c$ satisfies for all $i,j\in[m]$: 
    \begin{align*}
        \E_{(x,y)\sim\cD}[f(x)  y | x\in S_i]=\E_{(x,y)\sim\cD}[f(x) y | x\in S_j].
    \end{align*}
\end{claim}
The proof appears in Appendix \ref{app:equiv-proofs}. In the proof, we show that if $f$ does not satisfy the condition, we can construct a function $h\in\cF$ that is a group post-processing of $f$ and has a larger max-min value. 

When we have equal base rates, i.e. $\E[y|x\in S_i]$ is the same in all groups, equal opportunity also requires equal expected true positive counts. Therefore, in this setting maximizing the true positive count while satisfying equal opportunity is equivalent to satisfying max-min fairness. 

\begin{lemma}\label{lem:eo-maxmin}
    For $c<0.5$, let $\cD$ be a not degenerate distribution with respect to groups $S_1,\ldots,S_m$, capacity $2c$ and function family $\cF$ that is closed under group post-processing. Let $\cF_{EO}\subseteq\cF$ be the set of functions in $\cF$ satisfying equal opportunity.
    If the groups have equal base rates, $\E_{(x,y)\sim\cD}[y|x\in S_i] = \E_{(x,y)\sim\cD}[y]$ for all $i\in[m]$, then
    \begin{align*}
    \arg\max_{h\in\cF \text{ with capacity  }c}\min_{i\in[m]}\left\{\E_{(x,y)\sim\cD}[h(x) y | x\in S_i]\right\} \\= \arg\max_{h\in\cF_{EO}\text{ with capacity  }c}\E_{(x,y)\sim\cD}[h(x) y].
    \end{align*}
\end{lemma}
The proof appears in Appendix \ref{app:equiv-proofs}, where we show that both functions satisfy both definitions.
  
\paragraph{Non-equal base rates}
When the probability of $y=1$ differs between groups, equal opportunity is not equivalent to max-min fairness. We remark that if we would have defined the utility of a group $i$ as $\E_{(x,y)\sim\cD}[h(x) y|y=1]$, then these definitions would be identical even for non-equal base rates (assuming the non-degenerate condition). This difference is significant if there is a large gap in the base rates between groups.
When some group $S_i$ has higher base rate than others, $\E_{(x,y)\sim\cD}[y|x\in S_i]>\E_{(x,y)\sim\cD}[y|x\in S_j]$, equal opportunity allocates more resources to $S_i$ compared to max-min fairness. In equal opportunity we will have $\E_{(x,y)\sim\cD}[f(x) y|x\in S_i]>\E_{(x,y)\sim\cD}[f(x) y|x\in S_j]$ while in max-min fairness in a non-degenerate setting we have equality. 

\subsection{Proportional fairness and fairness regularization}
We show that the normalized proportional fairness expression, defined in Definition \ref{def:norm-prop}, can be written as a sum two terms. The first is the Kullback–Leibler divergence between two distributions, that is maximized for $f$ satisfying max-min fairness. The second term is maximized when $f$ maximizes the true positive count. Therefore, proportional fairness in the resource-limited setting can be viewed as a variation of fairness regularization.  
We remark that this is not a standard fairness regularization, as the $D_{KL}$ and log of the expected true positive count are  both logarithmic. This mean that this regularization is in fact multiplicative, rather than the standard additive regularization. 

We start by defining the two following distributions over $[m]$, the number of groups. One is of the group weights (independent of any allocation) and one that is the distribution of the true positive count of the groups for a specific function $f$.
\begin{definition}
    For a distribution $\cD$ on $\cX\times\{0,1\}$ and group partition $S_1,\ldots,S_m\subset \cX$, let $Q_\cS$ be the distribution over $[m]$ defined by, for all $i\in [m]$: \begin{align*}
        \Pr_{j\sim Q_\cS}[j=i]=\Pr_{(x,y)\sim\cD}[x\in S_i].
    \end{align*}
    Let $f:\cX\rightarrow[0,1]$ be any function, then the distribution $Q_f$ over $[m]$ is defined by, for all $i\in[m]$:
    \begin{align*}
        \Pr_{j\sim Q_f}[j=i]=\frac{\E_{(x,y)\sim\cD}[f(x)y\one(x\in S_i)]}{\E_{(x,y)\sim\cD}[f(x)y]}. 
    \end{align*} 
\end{definition}

We now can formalize the claim on proportional fairness.
\begin{lemma}\label{lemma:proportional-reg}
    For any distribution $\cD$ on $\cX\times\{0,1\}$, any partition of $\cX$ into groups $\cS = S_1,\ldots,S_m$ and any function $f:\cX\rightarrow[0,1]$, we have
    \begin{align*}
    \mathbf{Prop}_{\cD,\cS}(f) =
    -D_{KL}(Q_\cS,Q_f)+ \log \E_{(x,y)\sim\cD}[f(x) y].
\end{align*}
    
\end{lemma}
The proof appears on Appendix \ref{app:equiv-proofs}.
We prove that the $D_{KL}$ in our expression is indeed a fairness term, by showing that the function $f$ that minimizes it is the function satisfying max-min fairness.
\begin{claim}\label{claim:max-min-dkl}
Let $\cD$ be a non degenerate distribution with respect to groups $S_1,\ldots,S_m$, function family $\cF$ that is closed under group post-processing, and capacity $2c$ for $c<0.5$. The function $f\in\cF$ that has capacity $c$ and satisfies max-min fairness with respect to groups $S_1,\ldots,S_m$ also satisfies $D_{KL}(Q_\cS,Q_f)=0$.
\end{claim}
The proof appears in Appendix \ref{app:equiv-proofs}.
The claim implies that the term $D_{KL}(Q_\cS,Q_f)$ in the optimization is a form of distance from a function satisfying max-min fairness, as for any function $h$ we can write $D_{KL}(Q_\cS,Q_h)=D_{KL}(Q_f,Q_h)$, where $f$ is a function satisfying max-min fairness. When we have equal base rates, then this term can also be viewed as distance from equal opportunity.
\section{The price of fairness}\label{sec:price-of-fairness}
In this part we quantify the price of fairness in our setting. We prove the statements in this section only for two groups, and the lower bounds also apply in settings with more groups that have an information gap. This is because we can always unite different groups to two groups, and max-min fairness with respect to the union of groups is a weaker requirement than on the original groups. 

In this section we assume that our decision process is a two-step process, where in the first we find a predictor and in the second we post-process it. The information gap is modeled as the \emph{tail dominance} of $p$ on one group comparing to the other, as defined in Definition \ref{def:dominance}.

 When we post-process a predictor $p$ into a randomized allocation satisfying some requirement such as max-min fairness, equal opportunity or proportional fairness, we can have two different approaches to satisfy the requirements. If the predictor is calibrated with a small calibration error, one option is to assume that the output distribution of the predictor $p(x),x\sim\cD$, correctly characterized the distribution of $y$. In this case, we can enforce the fairness requirement on the \emph{simulated distribution}, as defined on Definition \ref{def:simulated}. 
We show that enforcing the fairness requirement on the simulated distribution implies that it is also satisfied, up to the an error, on the real distribution $\cD$. An advantage of this approach is that we can post-process a predictor $p$ to satisfy some fairness definitions without new labeled samples.

\subsection{Equal opportunity and max-min fairness}

For both max-min fairness and equal opportunity, we characterize the settings for which enforcing these definitions results in allocates most resources to the group we know less of. In the theorem we assume that we enforce the fairness definition on $\cD_p$, the simulated distribution (see Definition \ref{def:simulated}). In Appendix \ref{app:price} we also state prove a variant without this assumption, that has an additional error term.

\begin{theorem}\label{thm:eo-diff}
    Let $\cD$ be a distribution over $\cX\times\{0,1\}$ where $
    \cX$ is divided into $S_1,S_2$, that is $2c$ non-degenerate.
    Let $p:\cX\rightarrow[0,1]$ be an $\alpha$-calibrated predictor on $S_1,S_2$, such that group $S_1$ has $t_0$-tail dominance over $S_2$. Then for any $c\in[0,1]$ such that such that $\forall i\in[2], \E_{(x,y)\sim\cD}[t\geq t_0\one(x\in S_i)]\geq c$ and
     any randomized group threshold function $\tau:[0,1]\times[2]\rightarrow[0,1]$ on $p$ such that $\tau(p(x),g(x))$ satisfies max-min fairness and has capacity $c$ on the simulated distribution, also satisfies
    \begin{align*}
        \E_{(x,y)\sim\cD}[\tau(p(x),1)|x\in S_1]\leq \E_{(x,y)\sim\cD}[\tau(p(x),2)|x\in S_2].
    \end{align*}

\end{theorem}
The proof appears in Appendix \ref{app:price}.

We remark that using Lemma \ref{lem:eo-maxmin},
the same holds for equal opportunity with equal base rates, and when the true positive rate of $S_2$ is larger then of $S_1$. This theorem proves that when we have tail dominance, equal opportunity (with equal base rates) and max-min fairness are stricter than demographic parity. That is, they allocate the majority of the resources to the groups we have less information on.
For max-min fairness it is intuitive to see, as if we want the same utility for both groups and one can use the resources less efficiently, it needs to receive more resources. 

We prove next a gap version of the theorem, showing that gap tail dominance implies a gap on the allocated resources.  We first define a gap version of tail dominance. We remark that we only require a gap for a given range $[t_0,t_M]$, as it is a too strict definition to require a gap for all $r > t_0$. This is because in many of the risk-estimation settings we discuss there are no very high-risk individuals, for example there is no $x$ with $p(x)=0.9$. Therefore, it is unreasonable to require a gap for this range.
\begin{definition}\label{def:gap-dominance}
    Let $p:\cX\rightarrow[0,1]$ be a predictor, and let $S_1,S_2$ be two disjoint groups in $\cX$.
    Group $S_1$ has $t_0,t_M$-tail dominance over group $S_2$ with gap $\eta$ under distribution $\cD$ and predictor $p$ if for all $r\in [t_0,t_M]$ and $r'\geq t_M$
      \begin{align*}
        &\E_{(x,y)\sim\cD}[p(x)>r|x\in S_1] - \E_{(x,y)\sim\cD}[p(x)>r|x\in S_2] \geq \eta \\
        &\E_{(x,y)\sim\cD}[p(x)>r'|x\in S_1] \geq \E_{(x,y)\sim\cD}[p(x)>r'|x\in S_2]
    \end{align*}
\end{definition}
We prove a gap version of Theorem \ref{thm:eo-diff}, assuming tail dominance and a weak Lipchitz-type condition on the distribution of $p(x)$.
\begin{theorem}\label{thm:eo-diff-gap}
     Let $\cD$ be a distribution over $\cX\times\{0,1\}$ where $
    \cX$ is divided into $S_1,S_2$, and $\eta,\beta,t_0,t_m\in[0,1]$, and assume that it is $2c$ non-degenerate.
    Let $p:\cX\rightarrow[0,1]$ be an $\alpha$-calibrated predictor on $S_1,S_2$, such that group $S_1$ has $t_0,t_M$-tail dominance over $S_2$ with gap $\eta$ and for all $r\geq r_0$,
    $\E_{(x,y)\sim\cD}[p(x)\in[r,r+\beta]|x\in S_i]< \eta$, and $\Pr_{(x,y)\sim\cD}[p(x)\geq t_M]< \eta$.
    Let $c\in[6\eta,1]$ be a such that for $i\in[2]$, $\Pr_{(x,y)\sim\cD}[p(x)\geq t_0\wedge x\in S_i)]\geq c$. Then any randomized threshold function $\tau$ with capacity $c_0$, if $\tau(p(x),g(x))$ satisfies max-min fairness on $\cD_p$ then, 
    \begin{align*}
        \E_{(x,y)\sim\cD}[\tau(p(x),1)|x\in S_1]\leq \E_{(x,y)\sim\cD}[\tau(p(x),2)|x\in S_2] - \frac{\eta\beta}{t_2}.
    \end{align*}
\end{theorem}
The proof appears in Appendix \ref{app:price}.
This theorem shows that if the gap is large and the prediction on $S_2$ is very low (i.e. we barely have any information on individuals from $S_2$ and the probability for $y=1$ is low), then we end up allocating a strict majority of the resources to $S_2$. Such allocation leads to the true positive count going towards zero, as nearly all resources are allocated to individuals with $y=0$. We prove the following claim on the price of fairness, requiring a slightly stronger Lipchitz condition than in Theorem \ref{thm:eo-diff-gap}. In addition to the claim below, 
we demonstrate the price of fairness in Figure \ref{fig:price-large-gap}. The proof is in Appendix \ref{app:price}.
\begin{claim}\label{claim:additive-price-of}
    In the setting of Theorem \ref{thm:eo-diff-gap}, for equal-sized groups where $p(x)$ also satisfies $\E_{(x,y)\sim\cD}[p(x)\in[r,r+\beta]|x\in S_i]< \eta/4$ and $p(x)$ is a continuous distribution, the additive price of fairness is at least $\frac{\eta\beta}{8}+ \frac{\eta\beta^2}{4t_2}$.
\end{claim}
We remark that when $t_2$, the threshold for the second group, approaches zero then the price of fairness grows significantly. 

\begin{figure}[t]
    \centering
    \includegraphics[width=\columnwidth]{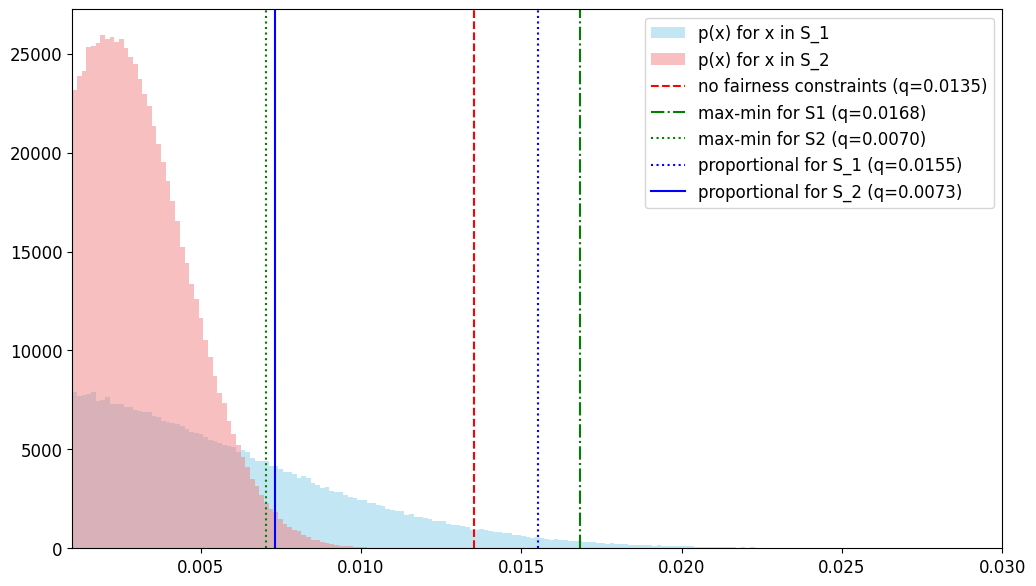}
    \caption{Simulating the output distribution of a predictor $p$ with a large gap between the distribution over $S_1$ and $S_2$. In such cases the true positive count with fairness requirements is significantly larger than without. See Appendix \ref{app:simulations} for more information.
    }
    \label{fig:price-large-gap}
\end{figure}

\subsection{Proportional fairness}
For proportional fairness we see next that the price of fairness is bounded. The harm to the true positive count is always at least half of the true positive count of the optimal solution without fairness requirements. Proportional fair solution might allocate most resources to the group we know less of, yet we are guaranteed a minimal usage of the resource. 

\begin{theorem}\label{thm:prop-bound}
      Let $\cD$ be a distribution over $\cX\times\{0,1\}$ where $\cX$ is divided into two equal-sized groups $S_1,S_2$. Let $p:\cX\rightarrow[0,1]$ be an $\alpha$-calibrated predictor and $\tau:[0,1]\times[2]\rightarrow[0,1]$ to be a randomized group threshold function such that $\tau(p(x),g(x))$ satisfies proportional fairness on the simulated distribution and has capacity $c$. Let $\tau'$ be the post-processing function such that $\tau'(p(x),g(x))$ maximizes the true positive count with capacity $c$.
    Then \begin{align*}
        \E_{(x,y)\sim\cD}[\tau(p(x),g(x))y]\geq \frac{1}{2}\E_{(x,y)\sim\cD}[\tau'(p(x),g(x))y]-\alpha
    \end{align*}
\end{theorem}
The proof appears in Appendix \ref{app:price}. In the proof we find what conditions the proportionally fair solution should satisfy, and using this we bound the price of fairness.

We remark that when $S_1,S_2$ are not equal-sized, then the factor $1/2$ is in fact the size of the dominant group. We further remark that this theory applies only for our setting. Proportional fairness for general utility function can have unbounded price of fairness.

\begin{figure}[t]
    \centering
    \includegraphics[width=\columnwidth]{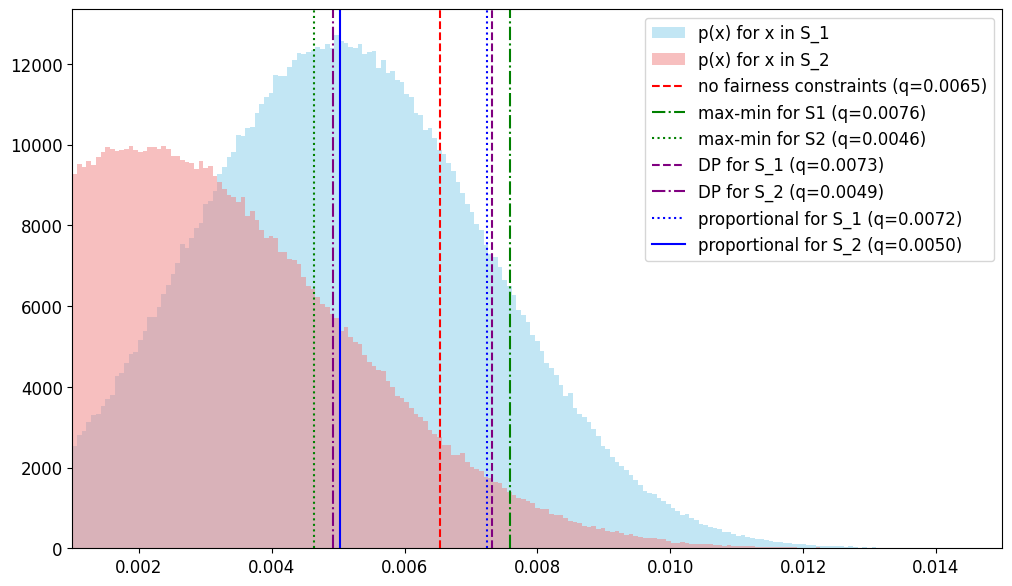}
    \caption{Simulating the output distribution for higher capacity under different type of tail dominance. See Appendix \ref{app:simulations}.
    }
    \label{fig:price-large-capacity}
\end{figure}

\section{New equal opportunity definition}\label{sec:new-def}
In this section we suggest a limited-capacity version of equal opportunity, that has a bounded price of fairness. Instead of equalizing the true positive rate, we only consider the individuals that both have positive outcome ($y=1$) and that  we could have reached if we would have allocated all resources to this group. We denote this group the \emph{achievable set}, and note that the achievable set depends both on the capacity and the hypothesis class.
\begin{definition}
    Let $\cD$ be a distribution over $\cX\times\{0,1\}$, $S_1,\ldots,S_m\subset\cX$ be a partition of $\cX$ and $\cF$ a set of functions.
    For every $c\in[0,1]$ let $\cF_c\subseteq\cF$ the set of functions in $\cF$ with capacity $c$.
    The achievable true positive count of $S_i$ with capacity $c$ is given by
    \begin{align*}
        \cA(\cF,c,S_i) = \max_{h\in\cF_c}\E_{(x,y)\sim\cD}[h(x)y|x\in S_i],
    \end{align*}
\end{definition}
Using the achievable set we define a new generalization of equal opportunity, that considers the achievable set instead of all of the $y=1$ in the group.
\begin{definition}[Achievable equal opportunity]
    Let $\cD$ be a distribution over $\cX\times\{0,1\}$, $S_1,\ldots,S_m\subset\cX$ be a partition of $\cX$ and $\cF$  a family of functions. A function $f:\cX\rightarrow[0,1]$ satisfies achievable equal opportunity if for all $i,j\in[m]$
    \begin{align*}
        \frac{\E[f(x)y|x\in S_i]}{\cA(\cF,c,S_i)} = \frac{\E[f(x)y|x\in S_j]}{\cA(\cF,c,S_j)}
    \end{align*}
\end{definition}
For $c=1$, if $\cF$ contains the constant $1$ function, then $\cA(\cF,c,S_i)=\E_{(x,y)\sim\cD}[y|x\in S_i]$ and the achievable equal opportunity requirement is exactly equal opportunity.

We prove next that the price of fairness is bounded for a rich enough class of functions $\cF$. This is a strengthening of Definition \ref{def:closed-post}, where we can choose a different function from $\cF$ for each group.
\begin{definition}
    A set of functions $\cF$ is closed under group choice post-processing if for every $f_1,\ldots,f_m\in\cF$ and post-processing function $\tau$, the function $f^*$ defined by $f^*(x)=\tau(f_{g(x)},g(x))$ is in $\cF$.
\end{definition}

\begin{claim}\label{claim:new-def}
Let $\cD$ be a distribution over $\cX\times\{0,1\}$, $S_1,\ldots,S_m\subset\cX$ be a partition of $\cX$ and $\cF$ be a function family closed under group choice post-processing. Then for every $c\in[0,1]$, if $f:\cX\rightarrow[0,1]$ has capacity $c$ and maximizes $\E_{(x,y)\sim\cD}[f(x)y]$ under achievable equal opportunity constraint, and let $h\in\cF$ be the function  maximizing $\E_{(x,y)\sim\cD}[h(x)y]$ without any fairness restrictions.
Then we have
\begin{align*}
    \E_{(x,y)\sim\cD}[f(x)y] \geq \frac{1}{m}\E_{(x,y)\sim\cD}[h(x)y].
\end{align*}
\end{claim}
We prove the claim in Appendix \ref{app:new-def}, by defining a new function with capacity $mc$ allocating capacity $c$ to every group. In the appendix we also show that the bound is tight by providing an example. 

\bibliography{refs}

@article{bertsimas2011price,
  title={The price of fairness},
  author={Bertsimas, Dimitris and Farias, Vivek F and Trichakis, Nikolaos},
  journal={Operations research},
  volume={59},
  number={1},
  pages={17--31},
  year={2011},
  publisher={INFORMS}
}

@book{dudley2018real,
  title={Real analysis and probability},
  author={Dudley, Richard M},
  year={2018},
  publisher={Chapman and Hall/CRC}
}

@misc{tsirelson_probability_notes,
  author       = {Tsirelson, Boris},
  title        = {Probability Theory},
  howpublished = {Lecture notes, Lecture 4},
  institution  = {Department of Mathematics},
  note         = {Tel Aviv University / Hebrew University of Jerusalem},
  url          = {https://www.ma.huji.ac.il/~ohadfeld/Tsirelson/Courses/Prob/lect4.pdf}
}

@article{hardt2016equality,
  title={Equality of opportunity in supervised learning},
  author={Hardt, Moritz and Price, Eric and Srebro, Nati},
  journal={Advances in neural information processing systems},
  volume={29},
  year={2016}
}

@inproceedings{hebert2018multicalibration,
  title={Multicalibration: Calibration for the (computationally-identifiable) masses},
  author={H{\'e}bert-Johnson, Ursula and Kim, Michael and Reingold, Omer and Rothblum, Guy},
  booktitle={International Conference on Machine Learning},
  pages={1939--1948},
  year={2018},
  organization={PMLR}
}

@inproceedings{hu2023omnipredictors,
  title={Omnipredictors for constrained optimization},
  author={Hu, Lunjia and Navon, Inbal Rachel Livni and Reingold, Omer and Yang, Chutong},
  booktitle={International Conference on Machine Learning},
  pages={13497--13527},
  year={2023},
  organization={PMLR}
}

@inproceedings{bin2018fairness,
  title={Fairness in resource allocation: Foundation and applications},
  author={Bin-Obaid, Hamoud S and Trafalis, Theodore B},
  booktitle={International Conference on Network Analysis},
  pages={3--18},
  year={2018},
  organization={Springer}
}

@article{zink2024race,
  title={Race adjustments in clinical algorithms can help correct for racial disparities in data quality},
  author={Zink, Anna and Obermeyer, Ziad and Pierson, Emma},
  journal={Proceedings of the National Academy of Sciences},
  volume={121},
  number={34},
  pages={e2402267121},
  year={2024},
  publisher={National Academy of Sciences}
}

@article{bird2025algorithms,
  title={Are algorithms biased in education? Exploring racial bias in predicting community college student success},
  author={Bird, Kelli A and Castleman, Benjamin L and Song, Yifeng},
  journal={Journal of Policy Analysis and Management},
  volume={44},
  number={2},
  pages={379--402},
  year={2025},
  publisher={Wiley Online Library}
}

@article{barth2024quantifying,
  title={Quantifying the extent to which junior performance predicts senior performance in Olympic sports: a systematic review and meta-analysis},
  author={Barth, Michael and G{\"u}llich, Arne and Macnamara, Brooke N and Hambrick, David Z},
  journal={Sports medicine},
  volume={54},
  number={1},
  pages={95--104},
  year={2024},
  publisher={Springer}
}

@article{ogryczak2014fair,
  title={Fair optimization and networks: A survey},
  author={Ogryczak, Wlodzimierz and Luss, Hanan and Pi{\'o}ro, Micha{\l} and Nace, Dritan and Tomaszewski, Artur},
  journal={Journal of Applied Mathematics},
  volume={2014},
  number={1},
  pages={612018},
  year={2014},
  publisher={Wiley Online Library}
}

@article{marsh1994equity,
  title={Equity measurement in facility location analysis: A review and framework},
  author={Marsh, Michael T and Schilling, David A},
  journal={European journal of operational research},
  volume={74},
  number={1},
  pages={1--17},
  year={1994},
  publisher={Elsevier}
}

@article{amanatidis2022fair,
  title={Fair division of indivisible goods: A survey},
  author={Amanatidis, Georgios and Birmpas, Georgios and Filos-Ratsikas, Aris and Voudouris, Alexandros A},
  journal={arXiv preprint arXiv:2202.07551},
  year={2022}
}

@inproceedings{donahue2020fairness,
  title={Fairness and utilization in allocating resources with uncertain demand},
  author={Donahue, Kate and Kleinberg, Jon},
  booktitle={Proceedings of the 2020 conference on fairness, accountability, and transparency},
  pages={658--668},
  year={2020}
}

@inproceedings{elzayn2019fair,
  title={Fair algorithms for learning in allocation problems},
  author={Elzayn, Hadi and Jabbari, Shahin and Jung, Christopher and Kearns, Michael and Neel, Seth and Roth, Aaron and Schutzman, Zachary},
  booktitle={Proceedings of the Conference on Fairness, Accountability, and Transparency},
  pages={170--179},
  year={2019}
}

@article{wang2018machine,
  title={A machine learning framework for resource allocation assisted by cloud computing},
  author={Wang, Jun-Bo and Wang, Junyuan and Wu, Yongpeng and Wang, Jin-Yuan and Zhu, Huiling and Lin, Min and Wang, Jiangzhou},
  journal={IEEE Network},
  volume={32},
  number={2},
  pages={144--151},
  year={2018},
  publisher={IEEE}
}

@article{gopalan2021omnipredictors,
  title={Omnipredictors},
  author={Gopalan, Parikshit and Kalai, Adam Tauman and Reingold, Omer and Sharan, Vatsal and Wieder, Udi},
  journal={arXiv preprint arXiv:2109.05389},
  year={2021}
}

@inproceedings{menon2018cost,
  title={The cost of fairness in binary classification},
  author={Menon, Aditya Krishna and Williamson, Robert C},
  booktitle={Conference on Fairness, accountability and transparency},
  pages={107--118},
  year={2018},
  organization={PMLR}
}

@inproceedings{corbett2017algorithmic,
  title={Algorithmic decision making and the cost of fairness},
  author={Corbett-Davies, Sam and Pierson, Emma and Feller, Avi and Goel, Sharad and Huq, Aziz},
  booktitle={Proceedings of the 23rd acm sigkdd international conference on knowledge discovery and data mining},
  pages={797--806},
  year={2017}
}

@article{donini2018empirical,
  title={Empirical risk minimization under fairness constraints},
  author={Donini, Michele and Oneto, Luca and Ben-David, Shai and Shawe-Taylor, John S and Pontil, Massimiliano},
  journal={Advances in neural information processing systems},
  volume={31},
  year={2018}
}

@inproceedings{diana2021minimax,
  title={Minimax group fairness: Algorithms and experiments},
  author={Diana, Emily and Gill, Wesley and Kearns, Michael and Kenthapadi, Krishnaram and Roth, Aaron},
  booktitle={Proceedings of the 2021 AAAI/ACM Conference on AI, Ethics, and Society},
  pages={66--76},
  year={2021}
}

@article{zhangproportional,
  title={Proportional Fairness in Federated Learning},
  author={Zhang, Guojun and Malekmohammadi, Saber and Chen, Xi and Yu, Yaoliang},
  journal={Transactions on Machine Learning Research},
year={2022}
}
\bibliographystyle{alpha}

\newpage
\appendix
\onecolumn
\section{Quantile Functions}\label{app:quant}
Some of the proofs in this paper use quantile functions and their properties. We define the basic definitions and prove some standard properties of the quantile functions below, for completeness.
\begin{definition}\label{def:CDF}
    Let $Q$ be a distribution over $[0,1]$, the cumulative distribution function (CDF) of $Q$ is the function $F:[0,1]\rightarrow[0,1]$ defined by
    \begin{align*}
         F(t) =  \Pr_{u\sim Q}[u\leq t].
    \end{align*}
\end{definition}
\begin{definition}\label{def:quantile}
    Let $Q$ be a distribution over $[0,1]$, the quantile function of $Q$ is the function $q:[0,1]\rightarrow[0,1]$  defined by 
    \begin{align*}
         q(\kappa) = \inf\left\{t\left| F(t)\geq\kappa \right.\right\}.
    \end{align*}
\end{definition}
This is the standard definition of a quantile function, although we remark that there are some cases in the literature that other definitions are sometimes used (where instead of the infimum, the supremum or other expression is used). 

We note that for every distribution $Q$, both the CDF $F$ and the quantile function $q$ are monotonically increasing. In addition, for every $t\in[0,1]$ we have that $q(F(t))\leq t$ (as we pick $q$ to be the infimum) and $F(q(k))\geq k$. We further note that for $K>F(t)$ we have that $q(K) > t$, because $q(K)=\inf\{r| F(r)\geq K\}$, since $F(t) < K$, $q(K)$ must be larger than $t$. 

\begin{claim}\label{claim:quantile-exp}
    Let $Q$ be a distribution over $[0,1],$ $q$ be the quantile function of it and $F$ be the commulative distribution function of $Q$. Then, for every $t<t_M,\in[0,1]$,
    \begin{align}
        &\E_{u\sim Q}[\one(u > t)u] = \int_{F(t)}^1 q(r)d r \label{eq:int-quant-unbounded}\\
        &\E_{u\sim Q}[\one(t<u\leq t_M)u] = \int_{F(t)}^{F(t_M)} q(r)dr\label{eq:int-quant-bounded}
    \end{align}
\end{claim}
\begin{proof}
    By definition, the quantile function is monotonically increasing and bounded, and therefore the integrals are well-defined.

    We start by proving (\ref{eq:int-quant-unbounded}). Let $v$ be the random variable defined by picking $u\sim Q$, then setting  
    $v=0$ if $u\leq t$ and $v=u$ otherwise. Let $Q_v$ be the distribution of $v$ and $q_v$ be the quantile function of $v$. 

    For $r\leq t$, we have that $\Pr_{v\sim Q_v}[v\leq r]=\Pr_{u\sim Q}[u\leq t]=F(t)$.
    This implies that for all $k< F(t)$ we have that $q_v(k)=\inf\left\{r\left| \Pr_{v\sim Q_v}[v\leq r]\geq k \right.\right\}=0$. 

    For every $r>t$ we have that $\Pr_{v\sim Q_v}[v\leq r]=\Pr_{u\sim Q}[u\leq r]=F(r)$. Therefore, for $k>F(t)$ we have that $q_v(k)=\inf\left\{r\left| \Pr_{v\sim Q_v}[v\leq r]\geq k \right.\right\}=\inf\left\{r\left| \Pr_{u\sim Q}[u\leq r]\geq k \right.\right\}=q(k)$.
    
    Combining it together we have: 
    \begin{align*}
        \E_{u\sim Q}[\one(u > t)u] =\E_{v\sim Q_v}[v]=\int_{0}^1 q_v(r)dr = \int_{F(t)}^1 q(r)dr,
    \end{align*}
    where we use the fact that we can express the expected value of any random variable as the integral over its quantiles, see \cite{dudley2018real,tsirelson_probability_notes}.

    The proof of (\ref{eq:int-quant-bounded}) follows the same approach. We define a new random variable $w$ by sampling $u\sim Q$ and then setting $w=0$ if $u\leq t$ or $u>t_M$ and $w=u$ otherwise. Let $Q_w$ be the distribution of this new random variable, and $q_w$ be it's quantile function.
    For $r\leq t$:
    \begin{align*}
        \Pr_{w\sim Q_w}[w\leq r]=\Pr_{u\sim Q}[u\leq t] + \Pr_{u\sim Q}[u> t_M]=F(t) + (1-F(t_M)).
    \end{align*}
    This implies that for all $k<F(t) + (1-F(t_M))$ we have $q_w(k)=0$. 

    For $t< r\leq t_M$ we have that $\Pr_{w\sim Q_w}[w\leq r]=\Pr_{u\sim Q}[u\leq r] + \Pr_{u\sim Q}[u> t_M] =F(r) + (1-F(t_M))$, which implies that for $F(t)+ (1-F(t_M))<k\leq F(t_M)+ (1-F(t_M))$ we have $q_w(k)=q(k-(1-F(t_M)))$.
    We note that $F(t_M)+ (1-F(t_M))=1$, so we have characterized $q_w$ on its full range.
    \begin{align*}
        \E_{u\sim Q}[\one(t<u\leq t_M)u]=\E_{w\sim Q_w}[w]=\int_0^1q_w(k)dk=
        \int_{F(t) + (1-F(t_M))}^{F(t_M) + (1-F(t_M))} q_w(k) dk =
        \int_{F(t)}^{F(t_M)} q(r) dr.
    \end{align*}
\end{proof}

\begin{claim}\label{claim:int-bound}
    Let $Q$ be a distribution over $[0,1]$, $F$ its CDF and $q$ its quantile function. Then for every $t\in[0,1],\phi\in[F(t),1]$,
    \begin{align*}
        \int_{F(t)}^\phi q(r)dr\geq (\phi-F(t))t.
    \end{align*}
\end{claim}
\begin{proof}
    If $\phi=F(t)$ then both sides equal zero and we are done. Therefore, we assume in the rest of the proof that $\phi > F(t)$, and let $\Delta=\phi-F(t)>0$.

    We first show that for all $\epsilon > 0$ we have that $q(F(t)+\epsilon) > t$. This is true since $q(F(t)+\epsilon)=\inf\{r| F(r)\geq F(t)+\epsilon\}$, since $F(t) < F(t)+\epsilon$ and $F$ is monotone, it must be that $q(F(t)+\epsilon)>t$. 

    Next, we bound the second half of the integral. Denote $\tilde{t}=q(F(t)+0.5\Delta)$, and from above we know that $\tilde{t} > t$. Then we have:
    \begin{align}\label{eq:top-half}
        \int_{F(t)+0.5\Delta}^{\phi}q(r)dr\geq 0.5\Delta \tilde{t} > 0.5\Delta t.
    \end{align}

    Assume towards a contradiction that the claim does not hold, that is $\int_{F(t)}^\phi q(r)dr<(\phi-F(t))t=\Delta t$. Using (\ref{eq:top-half}), this implies that
    \begin{align*}
        \int_{F(t)}^{F(t)+0.5\Delta}q(r)dr <  0.5\Delta t.
    \end{align*}
    Since $q$ is bounded, there must be $\epsilon$ such that $q(F(t)+\epsilon) < t$ and we reach a contradiction.

\end{proof}

\section{Proof of the preliminaries and definitions section}\label{app:def-proofs}

\begin{proof}[Proof of \ref{claim:existence-threshold}]
    Fix any $i\in[m]$, and we prove the claim on $S_i$. Let $Q_i$ be the distribution of $p(x)$ for $x\sim\cD$ such that $x\in S_i$.
    We denote by $F_i$ the CDF of $Q_i$ and by $q_i$ its quantile function (see definitions \ref{def:CDF} and \ref{def:quantile}).
    We define $t_i=q_i(1-c)$ and
\begin{align*}
    \gamma_i = \begin{cases}
        0   \quad &\Pr_{u\sim Q_i}[u=t_i]=0,\\
        \frac{c-1+F_i(t_i)}{\Pr_{u\sim Q_i}[u=t_i]} \quad &\text{o.w.}
    \end{cases}
\end{align*}
We show that $\gamma_i\in [0,1]$: since $t_i=q_i(1-c)$, we have that $F(t_i)\geq 1-c$, implying that $\gamma_i\geq 0$.
For the second bound, $F_i(t_i)=\Pr_{u\sim Q_i}[u\leq t_i] = \Pr_{u\sim Q_i}[u=t_i] + \Pr_{u\sim Q_i}[u < t_i]\leq \Pr_{u\sim Q_i}[u=t_i] + 1-c$, implying that $
\gamma_i\leq 1$.
The randomized threshold function $\tau$ on $S_i$ is defined by:
\begin{align*} 
    \tau(x,i) = \begin{cases}
        1 \quad &p(x)>t_i,\\
        0 \quad &p(x)<t_i\\
        \gamma_i \quad &p(x)=t_i.
    \end{cases}
\end{align*}
We show that $\tau$ satisfies the capacity requirements:
\begin{align*}
    \E_{x\sim\cD}[\tau(p(x),i)|x\in S_i]=&\E_{u\sim Q_i}[\tau(u,i)]\\=&
    1-F_i(q_i(1-c))+\Pr_{u\sim Q_i}[u=q_i(1-c)]\gamma_i\\=&
    1-F_i(q_i(1-c)) + (c-1+F_i(q_i(1-c))) = c.
\end{align*}
\end{proof}

\begin{proof}[Proof of Claim \ref{claim:loss-equiv}]

    Let $\ell:\{0,1\}\times\{0,1\}\rightarrow[0,1]$ be any bounded binary loss function with $\ell(0,0)=\ell(1,1)=0$. Let $h\in \cF$ be any function with capacity $c$.
    We express the expected loss of $h$, denoted by $\mathcal{L}_{\cD}(h)$: 
    \begin{align*}
        \mathcal{L}_{\cD}(h) =&\E_{(x,y)\sim\cD}[h(x)\ell(1,y) + (1-h(x))\ell(0,y)]\\
        =&\Pr_{(x,y)\sim\cD}[y=1]\E_{(x,y)\sim\cD}[(1-h(x))\ell(0,1)|y=1] + \Pr_{(x,y)\sim\cD}[y=0]\E_{(x,y)\sim\cD}[h(x)\ell(1,0)|y=0]
        \\=& \ell(0,1)\Pr_{(x,y)\sim\cD}[y=1]\left(1-\E_{(x,y)\sim\cD}[h(x)|y=1]\right) + \ell(1,0)\Pr_{(x,y)\sim\cD}[y=0]\E_{(x,y)\sim\cD}[h(x)|y=0]
    \end{align*}
    The function $h$ has capacity exactly $c$:
    \begin{align*}
        \Pr_{(x,y)\sim\cD}[y=0]\E_{(x,y)\sim\cD}[h(x)|y=0] = c-\Pr_{(x,y)\sim\cD}[y=1]\E_{(x,y)\sim\cD}[h(x)|y=1].
    \end{align*}
    We use this expression in the expected loss:
    \begin{align*}
        \mathcal{L}_{\cD}(h)=&\ell(0,1)\Pr_{(x,y)\sim\cD}[y=1]\left(1-\E_{(x,y)\sim\cD}[h(x)|y=1]\right) + \ell(1,0)\left(c-\Pr_{(x,y)\sim\cD}[y=1]\E_{(x,y)\sim\cD}[h(x)|y=1]\right)
        \\=&
        \ell(0,1)\Pr_{(x,y)\sim\cD}[y=1]+\ell(1,0)c-(\ell(1,0)+\ell(1,0))\E_{(x,y)\sim\cD}[h(x)y].
    \end{align*}
    The only term in $\mathcal{L}_{\cD}(h)$ that depends on $h$ is $-(\ell(1,0)+\ell(1,0))\E_{(x,y)\sim\cD}[h(x)y]$. 
    Therefore, any function $f\in\cF$ that maximizes $\E_{(x,y)\sim\cD}[f(x)y]$ is also the function that minimizes $\mathcal{L}_{\cD}(h)$.
\end{proof}

In the next claim we formalize the normalization requirement from Definition \ref{def:norm-prop}, that a uniform splitting of a group would not change the fairness requirement. This claim proves that the normalization we chose is the appropriate one for the fairness definition.
\begin{claim}
    Let $S_1,\ldots,S_m$ be a partition of $\cX$.
    For any $\eta\in(0,1)$, let $S^{(1)}_1,S^{(2)}_1$ be the groups created by splitting $S_1$ into two random subgroups of sizes $\eta,(1-\eta)$ respectively. Then for any function $f:\cX\rightarrow[0,1]$,
    \begin{align*}
        \mathbf{Prop}_{\cD,S_1,S_2,\ldots ,S_m}(f)=\mathbf{Prop}_{\cD,S_1^{(1)},S_1^{(2)},S_2,\ldots ,S_m}(f)
    \end{align*}
\end{claim}
\begin{proof}
The groups $S^{(1)}_1,S^{(2)}_1$ are random subsets of  $S_1$, so for every function $f:\cX\rightarrow[0,1]$
\begin{align*}
\E_{(x,y)\sim\cD,S_1^{(1)},S_1^{(2)}}\left[f(x) y | x\in S_1^{(1)}  \right] = \E_{(x,y)\sim\cD,S_1^{(1)},S_1^{(2)}}\left[f(x) y | x\in S_1^{(2)}  \right] = \E_{(x,y)\sim\cD}\left[f(x) y | x\in S_1  \right].
\end{align*}
By definition,
\begin{align*}
\mathbf{Prop}_{\cD,S_1^{(1)},S_1^{(2)},S_2,\ldots ,S_m}(f)=&
    \Pr_{(x,y)\sim\cD}\left[x\in S_1^{(1)}\right]\log\E_{(x,y)\sim\cD}\left[f(x) y | x\in S_1^{(1)}  \right] \\&+ \Pr_{(x,y)\sim\cD}\left[x\in S_1^{(2)}\right]\log\E_{(x,y)\sim\cD}\left[f(x) y | x\in S_1^{(2)}  \right]\\&+\sum_{i=2}^m\Pr_{(x,y)\sim\cD}[x\in S_i]\log\E_{(x,y)\sim\cD}\left[f(x) y | x\in S_{i}  \right]\\=&
    \left(\Pr_{(x,y)\sim\cD}\left[x\in S_1^{(1)}\right]+\Pr_{(x,y)\sim\cD}\left[x\in S_1^{(2)}\right]\right)\E_{(x,y)\sim\cD}\left[f(x) y | x\in S_1  \right]\\&+\sum_{i=2}^m\Pr_{(x,y)\sim\cD}[x\in S_i]\log\E_{(x,y)\sim\cD}\left[f(x) y | x\in S_{i}  \right]\\=&\mathbf{Prop}_{\cD,S_1,S_2,\ldots ,S_m}(f)
\end{align*}
\end{proof}

\section{Proof of Section \ref{sec:equiv}}\label{app:equiv-proofs}
\begin{proof}[Proof of Claim \ref{claim:max-min-equal}]
Assume towards a contradiction that we are in a $c$ non-degenerated distribution $\cD$, and that the function $f\in\cF$ satisfying max-min fairness does not have equal true positive count between groups. We find a new function $h\in\cF$ that is a post-processing of $f$ and has a higher max-min value, contradicting the optimality of $f$.

Let $i_1,\ldots,i_k\in [m]$ be all of the group in $[m]$ with a minimal true positive count, i.e. minimal $\mu_{\min}=\E_{(x,y)\sim\cD}[f(x)y|x\in S_{i_j}]$ (it is possible that $k=1$ and there is a single group with a minimal value). According to our assumption of unequal true positive count between groups, there exists at least one $i_{\max}\in [m]$ with $\mu_{\max}=\E_{(x,y)\sim\cD}[f(x)y|x\in S_{i_{\max}}] > \mu_{\min}$. 
Set
\begin{align*}
    &\delta = \min\left\{\frac{1}{2}\left(\mu_{\max}-\mu_{\min}\right),\Pr_{(x,y)\sim\cD}\left[x\in \cup_{j\in[k]}S_{i_j},f(x)\leq1-\mu_{\min}\right]\mu_{\min}\right\}\\
    &\epsilon = \delta \frac{\Pr_{(x,y)\sim\cD}[x\in S_{i_{\max}},f(x)\geq\delta]}{\Pr_{(x,y)\sim\cD}\left[x\in \cup_{j\in[k]}S_{i_j},f(x)\leq 1-\mu_{\min}\right]}
\end{align*}
We note that $\mu_{\min}< c< 0.5$ and so $1-\mu_{\min}>\mu_{\min}$. Combining with the fact that $f(x)\geq 0$ we get that $\Pr_{(x,y)\sim\cD}\left[x\in \cup_{j\in[k]}S_{i_j},f(x)<1-\mu_{min}\right] >0$ and $\epsilon$ is well-defined.

We define a new function $h$ that is a post-processing of $f$ by $h(x)=\tau(f(x),g(x))$ for: 
\begin{align*}
    \tau(r,i) = \begin{cases}
        r-\delta \quad &\text{if }i=i_{\max} \text{ and }r\geq\delta\\
        r+\epsilon \quad &\text{if }i\in \{i_1,\ldots,i_k\} \text{ and } r\leq 1-\mu_{\min}\\
        r \quad &\text{o.w.}.
    \end{cases}
\end{align*}
That is, we increase $f$ on $S_{i_j},\ldots,S_{i_j}$ and decrease its values on $S_{i_{\max}}$. The condition of $r\geq\delta$ is to ensure that the output of $h$ is in $[0,1]$ (as $\delta\in(0,1]$). For $\epsilon$, we need to show that $\epsilon\in[0,\mu_{\min}]$:

\begin{align*}
    \epsilon\leq\frac{\delta}{\Pr_{(x,y)\sim\cD}\left[x\in \cup_{j\in[k]}S_{i_j},f(x)<\mu_{max}\right]}\leq \mu_{\min} 
\end{align*}
by the definition of $\delta$, and therefore the output of $h$ is in $[0,1]$. 

We show next that $h$ has capacity exactly $c$:
\begin{align*}
    \E_{(x,y)\sim\cD}[h(x)]=&\sum_{i\in[m]}\Pr_{(x,y)\sim\cD}[x\in S_i]\E_{(x,y)\sim\cD}[h(x)|x\in S_i]\\
    =&\E_{(x,y)\sim\cD}[f(x)] - \delta \Pr_{(x,y)\sim\cD}[x\in S_{i_{\max}},f(x)\geq\delta] + \epsilon \Pr_{(x,y)\sim\cD}\left[x\in \cup_{j\in[k]}S_{i_j},f(x)\leq1-\mu_{\min}\right]\\
    &=\E_{(x,y)\sim\cD}[f(x)].
\end{align*}

To finish the proof, we show that $h$ has higher max-min value than the original $f$, contradicting it's optimality. We show it by going over all groups $i\in[m]$:
\begin{itemize}
    \item For $i_{\max}$: 
    $\E_{(x,y)\sim\cD}[h(x)y|x\in S_{i_{\max}}] \geq \E_{(x,y)\sim\cD}[f(x)y|x\in S_{i_{\max}}] - \delta > \mu_{\min}$, as $\delta <  \mu_{\max}-\mu_{\min}$. 
    \item For $i\notin \{i_1,\ldots,i_k\}\cup\{i_{\min}\}$: by definition $\E_{(x,y)\sim\cD}[h(x)y|x\in S_{i}] = \E_{(x,y)\sim\cD}[f(x)y|x\in S_{i}] > \mu_{\min}$.
    \item For $i\in \{i_1,\ldots,i_k\}$: Let $f_{out}$ be defined by $f_{out}(x)=\one(f(x)\leq 1-\mu_{\min})$, where $\one$ is the indicator function. 
    Then $f_{out}$ is a post-processing of $f$ and therefore in $\cF$. In addition,$c=\E_{(x,y)\sim\cD}[f(x)]\geq (1-\Pr_{(x,y)\sim\cD}[f(x)\leq 1-\mu_{\min}])(1-\mu_{\min})$ implying that $\Pr_{(x,y)\sim\cD}[f_{out}(x)=1] =\Pr_{(x,y)\sim\cD}[f(x)\leq 1-\mu_{\min}]\geq 1-\frac{c}{1-\mu_{\min}}\geq 1-2c$.
    From the non-degeneracy condition, $\E_{(x,y)\sim\cD}[y | x\in S_{i_j}, f_{out}(x)=1]>0$. From this we can bound:
    \begin{align*}
       \E_{(x,y)\sim\cD}[h(x)y|x\in S_{i_j}]=& \E_{(x,y)\sim\cD}[f(x)y|x\in S_{i_j}] + \epsilon \E_{(x,y)\sim\cD}[y\one(f_{out}(x)=1)|x\in S_{i_j}] \\>&
       \E_{(x,y)\sim\cD}[f(x)y|x\in S_{i_j}]
    \end{align*}
    where the inequality is because of the non degenerate condition, and the probability of $f(x)\leq 1-\mu_{\min}$ for $x\in S_{i_j}$ is strictly positive.
\end{itemize}

\end{proof}

  \begin{proof}[Proof of Lemma \ref{lem:eo-maxmin}]
    Let $f\in\cF$ be the function that satisfies max-min fairness with respect to $S_1,\ldots,S_m$, and let $h\in\cF$ be the function satisfying equal opportunity with respect to these groups and maximizing $\E_{(x,y)\sim\cD}[h(x)y]$, as in the Lemma statement. We prove the claim by showing the both $f,h$ satisfy both conditions, i.e. that $f$ satisfies equal opportunity and maximizes $\E_{(x,y)\sim\cD}[f(x)y]$, and that $h$ satisfies max-min fairness.

    First, we show that $f$ also satisfies the equal opportunity constraints. From Claim \ref{claim:max-min-equal}, we have that for all $i,j\in[m]$: $\E_{(x,y)\sim\cD}[f(x) y | x\in S_i]=\E_{(x,y)\sim\cD}[f(x) y | x\in S_j]$, and we denote this value by $\mu_f$. 
    Since we have equal base rates, for all $i,j\in[m]$
    \begin{align*}
        \E_{(x,y)\sim\cD}[f(x) |y=1, x\in S_i]=\frac{\mu_f}{\E_{(x,y)\sim\cD}[y=1| x\in S_i]}=\E_{(x,y)\sim\cD}[f(x) |y=1, x\in S_j].
    \end{align*}
    Which means that $f$ satisfies equal opportunity.

    For $h$, the equal opportunity constraints together with the equal base rates implies that for all $i,j\in[m]$:
    \begin{align*}
        \E_{(x,y)\sim\cD}[h(x) y|x\in S_i]=\E_{(x,y)\sim\cD}[y|x\in S_i]\E_{(x,y)\sim\cD}[h(x)|,y=1x\in S_i]=\E_{(x,y)\sim\cD}[h(x) y|x\in S_j].
    \end{align*}
    We denote this value by $\mu_h=\E_{(x,y)\sim\cD}[h(x) y|x\in S_i]$.
    We note that by the max-min definition, we have that $\min_{i\in[m]}\E_{(x,y)\sim\cD}[h(x)  y|x\in S_i]=\mu_h$.

    Using this notation
    \begin{align}\label{eq:max-min-util}
        \E_{(x,y)\sim\cD}[h(x)  y ]=\sum_{i\in[m]}\Pr_{(x,y)\sim\cD}[x\in S_i]\E_{(x,y)\sim\cD}[h(x) y |x\in S_i]=\sum_{i\in[m]}\Pr_{(x,y)\sim\cD}[x\in S_i]\mu_h=\mu_h.
    \end{align}
    Using an identical argument on $f$ we have that, $\E_{(x,y)\sim\cD}[f(x) y ]=\mu_f$.

    Since $h\in\cF$ is the function in $\cF$ that satisfies equal opportunity and maximizes $\E_{(x,y)\sim\cD}[f(x)  y ]$ and $f$ satisfies equal opportunity, we have that
    \begin{align*}
        \E_{(x,y)\sim\cD}[h(x) y ] = \mu_h \geq \E_{(x,y)\sim\cD}[f(x) y ] = \mu_f.
    \end{align*}

    The function $f$ is satisfies max-min fairness, therefore
    \begin{align*}
        \min_{i\in[m]}\E_{(x,y)\sim\cD}[h(x) y|x\in S_i]=\mu_h\leq \min_{i\in[m]}\E_{(x,y)\sim\cD}[f(x) y|x\in S_i]=\mu_f.
    \end{align*}
    We conclude that $\mu_f=\mu_h$, and that both functions satisfy both requirements.
\end{proof}

\begin{proof}[Proof of Lemma \ref{lemma:proportional-reg}]
We write the proportional fairness optimization expression explicitly:
\begin{align*}
    \mathbf{Prop}_{\cD,\cS}(f) =&
    \sum_{i\in[m]}\Pr_{(x,y)\sim\cD}[x\in S_i]\log \E_{(x,y)\sim\cD}[f(x) y | x\in S_i]\\=&
    \sum_{i\in[m]}\Pr_{(x,y)\sim\cD}[x\in S_i]\log \frac{\E_{(x,y)\sim\cD}[f(x) y \one(x\in S_i)]}{\Pr_{(x,y)\sim\cD}[x\in S_i]}\\=&
    \sum_{i\in[m]}\Pr_{(x,y)\sim\cD}[x\in S_i]\log \frac{\E_{(x,y)\sim\cD}[f(x) y\one(x\in S_i)]\E_{(x,y)\sim\cD}[f(x) y] }{\Pr_{(x,y)\sim\cD}[x\in S_i]\E_{(x,y)\sim\cD}[f(x) y]}\\=&
    \sum_{i\in[m]}\Pr_{(x,y)\sim\cD}[x\in S_i]\left(\log \frac{\E_{(x,y)\sim\cD}[f(x) y \one(x\in S_i)]}{\Pr_{(x,y)\sim\cD}[x\in S_i]\E_{(x,y)\sim\cD}[f(x) y]} + \log \E_{(x,y)\sim\cD}[f(x)y]\right) \\=&
    -D_{KL}(Q_\cS,Q_f)+ \log \E_{(x,y)\sim\cD}[f(x)  y] 
\end{align*}
Where the minus sign comes from flipping the denominator and numerator to get the divergence.
\end{proof}
\begin{proof}[Proof of Claim \ref{claim:max-min-dkl}]
    From Claim \ref{claim:max-min-equal}, a function $f$ satisfying max-min fairness requirement for a non-degenerate distribution satisfies that for all $i\neq j$:
    \begin{align*}
        \E_{(x,y)\sim\cD}[f(x)y|x\in S_i]=\E_{(x,y)\sim\cD}[f(x)y|x\in S_j] = \E_{(x,y)\sim\cD}[f(x)y].
    \end{align*}
    Such $f$ satisfies for every $i\in[m]$:
    \begin{align*}
        \Pr_{j\sim Q_f}[j=i]=\frac{\E_{(x,y)\sim\cD}[f(x)y\one(x\in S_i)]}{\E_{(x,y)\sim\cD}[f(x)y]} = \E_{(x,y)\sim\cD}[\one(x\in S_i)] = \Pr_{j\sim Q_\cS}[j=i].
    \end{align*}
    Which implies that $D_{KL}(Q_\cS,Q_f)=0$.
\end{proof}

\section{Proof of Section \ref{sec:price-of-fairness}}\label{app:price}

In all of the claims, we have assumed the predictor $p$ is perfectly calibrated. Using the following claim, we show that the claims still holds where $p$ has small calibration error, where we pay a factor of our calibration error.
\begin{claim}\label{claim:calibration}
    Let $\cD$ be a distribution over $\cX\times\{0,1\}$ where $
    \cX$ is divided into groups $S_1,\ldots,S_m$. 
    Let $p:\cX\rightarrow[0,1]$ be a predictor that is $\alpha$-calibrated on $S_1,S_2$, satisfying for all $i\in[m]$
    \begin{align}
        \int_0^1 \left|\E_{(x,y)\sim\cD}[(y-r)\wedge \one(p(x)=r)|x\in S_i]\right|dr \leq\alpha
    \end{align}
    Then, for every randomized group threshold function $\tau:[0,1]\times[m]\rightarrow[0,1]$ and every group $i$, if $f:\cX\rightarrow[0,1]$ is defined by $f(x)=\tau(p(x),g(x))$ for all $x$, then
    \begin{align}
        \left| \E_{(x,y)\sim\cD}[\tau(p(x),i)y|x\in S_i] -  \E_{(x,y)\sim\cD}[\tau(p(x),i) p(x)|x\in S_i] \right| \leq \alpha.
    \end{align}
\end{claim}

\begin{proof}
    Fix some $i\in[m]$ and let $t_i,\gamma_i$ be the threshold and value at the threshold of $\tau$. That is, $\tau(t_i,i)=\gamma_i$ and for all $r>t_i$, $\tau(r,i)=1$.
    By definition,
    \begin{align*}
        \E_{(x,y)\sim\cD}[\tau(p(x),i)y|x\in S_i]
        =& \gamma_i\E_{(x,y)\sim\cD}[\one(p(x)=t_i) y|x\in S_i] + \E_{(x,y)\sim\cD}[\one(p(x)>t_i) y|x\in S_i].
    \end{align*}
    Then we can bound the error in approximation, using the triangle inequality
    \begin{align*}
        \left| \E_{(x,y)\sim\cD}[\tau(p(x),i)y|x\in S_i]-\E_{(x,y)\sim\cD}[\tau(p(x),i) p(x)|x\in S_i] \right| \leq& \gamma_i  \left|\E_{(x,y)\sim\cD}[\one(p(x)=t_i) (y-p(x))|x\in S_i] \right| \\&+  \left|\E_{(x,y)\sim\cD}[\one(p(x)>t_i) (y-p(x))|x\in S_i] \right|
        \\\leq & \int_{t_i}^1 \left|\E_{(x,y)\sim\cD}[(y-r)\wedge \one(p(x)=r)|x\in S_i] \right| \leq \alpha.
    \end{align*}
\end{proof}

\begin{proof}[Proof of Theorem \ref{thm:eo-diff}]
    Let $\cD,p$ be as in the theorem statement. Since $\tau$ satisfies max-min fairness with respect to groups $S_1,S_2$ on the simulated distribution and we are in a non degenerate setting, it satisfies $\E_{(x,y)\sim\cD_p}[\tau(p(x),1)p(x)|x\in S_1]= \E_{(x,y)\sim\cD_p}[\tau(p(x),2)p(x)|x\in S_2]$.
    Since $\tau$ has capacity $c$ and $t_0$ is such that $\E_{(x,y)\sim\cD}[t\geq t_0 \one(x\in S_i)]\geq c$, it must be that both thresholds in $\tau$ are at least $t_0$ (else we are over our capacity).
    We apply Claim \ref{claim:max-min-dp} with $\alpha=0$, and get that
    \begin{align*}
        \E_{(x,y)\sim\cD}[\tau(p(x),1)|x\in S_1]\leq \E_{(x,y)\sim\cD}[\tau(p(x),2)|x\in S_2].
    \end{align*}
    as required.

\begin{claim}\label{claim:max-min-dp}
    Let $\cD$ be a distribution over $\cX\times\{0,1\}$ where $
    \cX$ is divided into two groups $S_1,S_2$. 
    Let $p:\cX\rightarrow[0,1]$ be a predictor on $S_1,S_2$, such that $S_1$ is $t_0$-tail dominant over $S_2$.
   
    For any randomized group threshold function $\tau:[0,1]\times[2]\rightarrow[0,1]$ with thresholds $t_1,t_2\geq t_0$, if $\tau$ on $p$ satisfies 
    \begin{align}
        \Pr_{(x,y)\sim\cD}[\tau(p(x),1)p(x)|x\in S_1] \leq\Pr_{(x,y)\sim\cD}[\tau(p(x),2)p(x)|x\in S_2] + \alpha \label{eq:equal-p-exp}
    \end{align}
    then it also satisfies
    \begin{align*}
    \Pr_{(x,y)\sim\cD}[\tau(p(x),1)|x\in S_1] \leq\Pr_{(x,y)\sim\cD}[\tau(p(x),2)|x\in S_2] + \frac{\alpha}{t_{\min}},
    \end{align*}
    where $t_{\min}=\min\{t_1,t_2\}$.
\end{claim}
\begin{proof}
    Let $\tau:[0,1]\times[2]\rightarrow[0,1]$ be a post-processing function satisfying the conditions of the claim. Let $t_1,t_2$ and $\gamma_1,\gamma_2$ be the thresholds and values of $\tau$ on the threshold, for groups $S_1,S_2$ respectively.

    We start by defining the distribution $Q_i$, which is be the distribution of $p(x)$ when picking $x\sim\cD$ conditioning on $x\in S_i$. That is, $\Pr_{u\sim Q_i}[u=r]=\Pr_{(x,y)\sim\cD}[p(x)=r|x\in S_i]$.
    Let $F_i$ be the CDF of this distribution and $q_i$ be the quantile function, see definitions in Appendix \ref{app:quant}. Using these notations and Claim \ref{claim:quantile-exp}:
    \begin{align*}
        \E_{(x,y)\sim\cD}[\tau(p(x),i)p(x)|x\in S_i]=&\gamma_i\E[\one(p(x)=t_i)p(x)|x\in S_i] + \E_{(x,y)\sim\cD}[\one(p(x)>t_i)p(x)|x\in S_i]\\=&
        \gamma_i\E[\one(p(x)=t_i)p(x)|x\in S_i] + \int_{F_i(t)}^1 q_i(\kappa)d\kappa.
    \end{align*}

    Next, we show that for all $\kappa\in(F_1(t_0),1]$, $q_1(\kappa)\geq q_2(\kappa)$.
    
    Assume towards a contradiction that there exists $\kappa\in(F_1(t_0),1]$ such that $q_1(\kappa) < q_2(\kappa)$, and let $t = 0.5(q_1(\kappa) + q_2(\kappa))$. We have that $t_0 < q_1(\kappa)< t < q_2(\kappa)$ (as $\kappa>F_1(t_0)$, see Definition \ref{def:quantile}).
    For this $t$, by the definition of the quantile function
    \begin{align*}
        &\Pr_{(x,y)\sim\cD}[p(x)> t|x\in S_1] \leq \Pr_{(x,y)\sim\cD}[p(x) > q_1(\kappa)|x\in S_1] =1-F(q_1(\kappa)) \leq 1-\kappa,\\
        &\Pr_{(x,y)\sim\cD}[p(x) > t|x\in S_2] = 1-F_2(t) > 1-\kappa, 
    \end{align*}
    as if $F_2(t)\geq \kappa$, we would have $q_2(\kappa)\leq t$ contradicting our definition above. The two statements together imply that $\Pr_{(x,y)\sim\cD}[p(x)> t|x\in S_1]<\Pr_{(x,y)\sim\cD}[p(x)> t|x\in S_2]$, contradicting the tail dominance of $S_1$.

    We use the inequality on the quantile functions to prove the claim, by dividing into a few different cases. The first case is in fact not necessary (covered by the other cases) but as it is the simplest we keep for clarity. The first case is where $p(x)$ is continuous and the high-level idea of the proof appears in it.
    \begin{enumerate}
        \item If $\Pr_{(x,y)\sim\cD}[p(x)=t_1|x\in S_1]=\Pr_{(x,y)\sim\cD}[p(x)=t_2|x\in S_2]=0$: This is the simplest case, where we have no weight on the thresholds. In this case, $\E_{(x,y)\sim\cD}[\tau(p(x),i)p(x)|x\in S_i]= \int_{F_i(t_i)}^1 q_i(\kappa)d\kappa$ and $\E_{(x,y)\sim\cD}[\tau(p(x),i)|x\in S_i]=1-F_i(t_i)$.
        Then, (\ref{eq:equal-p-exp}) implies
        \begin{align*}
            \int_{F_1(t_1)}^1 q_1(\kappa)d\kappa \leq \int_{F_2(t_2)}^1 q_2(\kappa)d\kappa +\alpha
        \end{align*}
        If $F_1(t_1) \geq F_2(t_2)$, we are done. Else, we have
        \begin{align*}
            \int_{F_1(t_1)}^1 q_1(\kappa)d\kappa=&\int_{F_1(t_1)}^{F_2(t_2)} q_1(\kappa)d\kappa+\int_{F_2(t_2)}^1 q_1(\kappa)d\kappa \\\geq&
            \int_{F_1(t_1)}^{F_2(t_2)} q_1(\kappa)d\kappa+\int_{F_2(t_2)}^1 q_2(\kappa)d\kappa,
        \end{align*}
        where we use the fact that $q_1(\kappa)\geq q_2(\kappa)$ for all $\kappa$ in the range.
        From both inequalities we conclude that $\int_{F_1(t_1)}^{F_2(t_2)} q_1(\kappa)d\kappa\leq\alpha$, implying that $(F_2(t_2)-F_1(t_1))t_1\leq\alpha$, where we use Claim \ref{claim:int-bound} to lower bound the integral.
        Therefore, we have that $F_1(t_1)\geq F_2(t_2) - \alpha/t_1$.

    \item{If $t_1<t_2$}: set $t=0.5(t_1+t_2)$, then $t_0\leq t_1 < t < t_2$. We express the true positive count using $t$:
        \begin{align*}
            \E_{(x,y)\sim\cD}[\tau(p(x),1)p(x)|x\in S_1]=&\E_{(x,y)\sim\cD}[\one(p(x)=t_1)\gamma_1 t_1|x\in S_1] +\E_{(x,y)\sim\cD}[\one(p(x)>t_1)p(x)|x\in S_1]\\=
            &\E_{(x,y)\sim\cD}[\one(p(x)=t_1)\gamma_1 t_1|x\in S_1] +\E_{(x,y)\sim\cD}[\one(t_1 < p(x)\leq t)p(x)|x\in S_1] \\&+ \E_{(x,y)\sim\cD}[\one(p(x)> t)p(x)|x\in S_1]\\
           \E_{(x,y)\sim\cD}[\tau(p(x),2)p(x)|x\in S_2]=&\E_{(x,y)\sim\cD}[\one(p(x)=t_2)\gamma_2 t_2|x\in S_2] +\E_{(x,y)\sim\cD}[\one(p(x)>t_2)p(x)|x\in S_2]\\=&
           \E_{(x,y)\sim\cD}[\one(p(x)>t)p(x)|x\in S_2] - \E_{(x,y)\sim\cD}[\one(   t_2<p(x)<t)p(x)|x\in S_2] \\&-
           \E_{(x,y)\sim\cD}[\one(p(x)=t_2)t_2(1-\gamma_2)|x\in S_2]
        \end{align*}
        Since $q_1(\kappa)\geq q_2(\kappa)$ for all $\kappa$ in range, we have that
        \begin{align*}
           \E_{(x,y)\sim\cD}[\one(p(x)>t)p(x)|x\in S_2]= \int_{F_2(t)}^1 q_2(\kappa)d\kappa\leq \int_{F_1(t)}^1 q_1(\kappa)d\kappa=\E_{(x,y)\sim\cD}[\one(p(x)>t)p(x)|x\in S_1].
        \end{align*}
        Therefore, in order to satisfy (\ref{eq:equal-p-exp}) we must have:
        \begin{align*}
            \alpha\geq& \E_{(x,y)\sim\cD}[\one(t_1 < p(x)\leq t)p(x)|x\in S_1] + \gamma_1\E_{(x,y)\sim\cD}[\one(p(x)= t_1)t_1|x\in S_1]\\&+
            \E_{(x,y)\sim\cD}[\one(   t_2<p(x)<t)p(x)|x\in S_2] +
           \E_{(x,y)\sim\cD}[\one(p(x)=t_2)t_2(1-\gamma_2)|x\in S_2] \\\geq&
           \int_{F_1(t_1)}^{F_1(t)} q_1(\kappa)d\kappa + \int_{F_2(t)}^{F_2(t_2)} q_2(\kappa)d\kappa + t_1\E_{(x,y)\sim\cD}[\one(p(x)=t_1)\gamma_1|x\in S_1]\\&-\E_{(x,y)\sim\cD}[\one(p(x)=t_2)\gamma_2|x\in S_2]t_1
        \end{align*}
        We lower bound the integrals by $(F_1(t)-F_1(t_1))t_1 + (F_2(t_2)-F_2(t))t_2$ using Claim \ref{claim:int-bound}. Since we have that $t_1<t_2$ and $F_2(t) \geq F_1(t)$, we can lower bound the integrals by $(F_2(t_2)-F_1(t_1))t_1$.

        Therefore, we conclude that
        \begin{align*}
            \alpha\geq& (F_2(t_2)-F_1(t_1))t_1+t_1\gamma_1\Pr_{(x,y)\sim\cD}[p(x)=t_1|x\in S_1]-\gamma_2\Pr_{(x,y)\sim\cD}[p(x)=t_2|x\in S_2]t_1\\\geq& \E_{(x,y)\sim\cD}[\tau(p(x),1)|x\in S_1]t_1-\E_{(x,y)\sim\cD}[\tau(p(x),2)|x\in S_2]t_1,
        \end{align*}
        finishing the claim for this case.
    
        \item If $t_1\geq t_2, F_1(t_1)\leq F_2(t_2)$: Denote $\delta = F_2(t_2)-F_1(t_1)$. 
        \begin{align*}
        \E_{(x,y)\sim\cD}[\tau(p(x),1)p(x)|x\in S_1]= &\E_{(x,y)\sim\cD}[\one(p(x)=t_1)\gamma_1 t_1|x\in S_1] + \int_{F_1(t_1)}^1 q_1(\kappa)d\kappa \\=&
        \E_{(x,y)\sim\cD}[\one(p(x)=t_1)\gamma_1 t_1|x\in S_1] + \int_{F_1(t_1)}^{F_2(t_2)} q_1(\kappa)d\kappa+ \int_{F_2(t_2)}^{1} q_1(\kappa)d\kappa\\
        \E_{(x,y)\sim\cD}[\tau(p(x),1)p(x)|x\in S_2]= &\E_{(x,y)\sim\cD}[\one(p(x)=t_2)\gamma_2 y|x\in S_2] + \int_{F_2(t_2)}^1 q_2(\kappa)d\kappa
    \end{align*}
    We know that $q_1(\kappa)\geq q_2(\kappa)$ for all $\kappa$ in our range. In order to satisfy (\ref{eq:equal-p-exp}):
    \begin{align*}
        \alpha + \E_{(x,y)\sim\cD}[\one(p(x)=t_2)\gamma_2|x\in S_2]t_2\geq& \E_{(x,y)\sim\cD}[\one(p(x)=t_1)\gamma_1 |x\in S_1]t_1 +\int_{F_1(t_1)}^{F_2(t_2)} q_1(\kappa)d\kappa\\\geq&
        \E_{(x,y)\sim\cD}[\one(p(x)=t_1)\gamma_1 |x\in S_1]t_1 +\delta t_1,
    \end{align*}
    where the last inequality is from Claim \ref{claim:int-bound}.
    Since $t_1\geq t_2$,
    which means 
    \begin{align*}
        \E_{(x,y)\sim\cD}[\one(p(x)=t_2)\gamma_2|x\in S_2]\geq \E_{(x,y)\sim\cD}[\one(p(x)=t_1)\gamma_1|x\in S_1]+ \delta - \frac{\alpha}{t_2}.
    \end{align*}
    which is enough to finish the claim in this case.
    \item If $t_1\geq t_2, F_1(t_1)> F_2(t_2)$: Denote $\Delta = F_1(t_1)-F_2(t_2)$.
        \begin{align*}
        \E_{(x,y)\sim\cD}[\tau(p(x),1)p(x)|x\in S_1]= &\E_{(x,y)\sim\cD}[\one(p(x)=t_1)\gamma_1 y|x\in S_1] + \int_{F_1(t_1)}^1 q_1(\kappa)d\kappa\\
        \E_{(x,y)\sim\cD}[\tau(p(x),1)p(x)|x\in S_2]=&\E_{(x,y)\sim\cD}[\one(p(x)=t_2)\gamma_2 y|x\in S_2] + \int_{F_2(t_2)}^1 q_2(\kappa)d\kappa \\=&
        \E_{(x,y)\sim\cD}[\one(p(x)=t_2)\gamma_2 y|x\in S_2] + \int_{F_2(t_2)}^{F_1(t_1)} q_2(\kappa)d\kappa + \int_{F_1(t_1)}^1 q_2(\kappa)d\kappa.
    \end{align*}
    \end{enumerate}
    To satisfy (\ref{eq:equal-p-exp}) we must have:
    \begin{align*}
        \E_{(x,y)\sim\cD}[\one(p(x)=t_1)\gamma_1 |x\in S_1]t_1\leq &\E_{(x,y)\sim\cD}[\one(p(x)=t_2)\gamma_2 y|x\in S_2]t_2 + \int_{F_2(t_2)}^{F_1(t_1)} q_2(\kappa)d\kappa+\alpha\\\leq&
        \E_{(x,y)\sim\cD}[\one(p(x)=t_2)\gamma_2 y|x\in S_2]t_2 +\Delta t_2 + \alpha
    \end{align*}
    Since $t_1\geq t_2$ we have that $\E_{(x,y)\sim\cD}[\one(p(x)=t_1)\gamma_1 |x\in S_1]\leq \E_{(x,y)\sim\cD}[\one(p(x)=t_2)\gamma_2 y|x\in S_2]+\Delta+\alpha/t_1$, which finishes the proof.
\end{proof}

\end{proof}

\begin{proof}[Proof of Theorem \ref{thm:eo-diff-gap}]
    Let $p,\tau,\cD$ as in the theorem statement, and denote by $t_1,t_2$ the thresholds of $\tau$ on groups $S_1,S_2$ respectively. 
    Group $S_1$ has $t_0,t_M$ tail dominance over $S_1$, and therefore it satisfies conditions (\ref{eq:dominance}),(\ref{eq:weak-dominance}) of Claim \ref{claim:y-exp-gap}. Since $\tau(p(x),g(x))$ has capacity $c$, from the conditions of the claim we have that $t_1,t_2\geq t_0$. From the theorem statement, $p(x)$ also satisfies (\ref{eq:small-tail}) and the Lipchitz-type condition (\ref{eq:lipschitz}).

    Therefore, $\tau$ satisfies the conditions of \ref{claim:y-exp-gap}, and we finish the proof.
\end{proof}

\begin{claim}\label{claim:y-exp-gap}
    Let $\cD$ be a distribution over $\cX\times\{0,1\}$ where $
    \cX$ is divided into two groups $S_1,S_2$. 
    Let $p:\cX\rightarrow[0,1]$ be a predictor on $S_1,S_2$.
    Let $t_0,t_M,\eta,\beta\in(0,1)$ be constants such that:
    \begin{align}
        &\E_{(x,y)\sim\cD}[p(x)\geq t_M|x\in S_2] \leq \eta \label{eq:small-tail}\\
        \forall r\in [t_0,t_M]\quad&\E_{(x,y)\sim\cD}[p(x)>r|x\in S_1] - \E_{(x,y)\sim\cD}[p(x)>r|x\in S_2] \geq \eta \label{eq:dominance}\\
        \forall r\geq t_M\quad&\E_{(x,y)\sim\cD}[p(x)>r|x\in S_1] \geq \E_{(x,y)\sim\cD}[p(x)>r|x\in S_2]\label{eq:weak-dominance}\\
        \forall r\geq t_0\quad&\E_{(x,y)\sim\cD}[p(x)\in [r,r+\beta]|x\in S_i] < \eta \label{eq:lipschitz}
    \end{align}
    
    For every randomized threshold function $\tau:\cX\times[2]\rightarrow[0,1]$ with thresholds $t_1,t_2>t_0$ and capacity $c\geq 6\eta$
    satisfies 
    \begin{align}
        \Pr_{(x,y)\sim\cD}[\tau(p(x),1)p(x)|x\in S_1] \leq\Pr_{(x,y)\sim\cD}[\tau(p(x),2)p(x)|x\in S_2]\label{eq:equal-p-exp-2}
    \end{align}
    Then it also satisfies 
    \begin{align}
        \E_{(x,y)\sim\cD}[\tau(p(x),2)|x\in S_2] - \E_{(x,y)\sim\cD}[\tau(p(x),1)|x\in S_1] \geq \frac{\beta\eta}{t_2}.\label{y-exp-gap}
    \end{align}
\end{claim}

\begin{proof}
    The proof of this claim follows the same high-level approach of Claim \ref{claim:max-min-dp}, but preserves the gap using the additional assumption on the density of the distribution defined by $p(x),x\sim\cD$. Let $t_1,t_2,\gamma_1,\gamma_2$ be the thresholds and values of $\tau$. 

    Following the approach of the previous proof, let $F_i$ be the comulative distribution function of $p(x)$ for $x\sim\cD$ conditioning on $x\in S_i$, and let $q_i$ be the quantile function of this distribution. 
    By the definition of $\tau$ and Claim \ref{claim:int-bound},\begin{align*}
        \E_{(x,y)\sim\cD}[\tau(p(x),i)p(x)|x\in S_i]=\Pr_{(x,y)\sim\cD}[p(x)=t_i|x\in S_i]\gamma_i\cdot t_i + \int_{F_i(t_i)}^1 q_i(\kappa)d\kappa.
    \end{align*}

    Let $K_0=F_1(t_0)$. We start by showing that for all $\kappa\in (K_0,1-2\eta)$ we have that
    $q_1(\kappa)\geq q_2(\kappa) + \beta$. 
    
    Assume towards a contradiction that there is some $\kappa\in(K_0,1-2\eta)$ such that $q_1(\kappa)< q_2(\kappa) + \beta$, and set $t=0.5(q_1(\kappa)+ q_2(\kappa) + \beta)$.
    Then $t >  q_1(\kappa)\geq t_0$. We show that $t<t_M$: if $t\geq t_M$ then $F_2(t)\geq F_2(t_M)\geq 1-\eta$. But, we also have that $F_2(t)\leq F_2(t-\beta) + \Pr_{(x,y)\sim\cD}[p(x)\in [t-\beta,t]]< 1-2\eta + \eta$, where we use (\ref{eq:lipschitz}) and the fact that $t-\beta < q_2(1-2\eta)$.
    Therefore, we can use (\ref{eq:dominance}) on $t$:
    \begin{align*}
        &\Pr_{(x,y)\sim\cD}[p(x) > t|x\in S_1] =\Pr_{(x,y)\sim\cD}[p(x) > q_1(\kappa)|x\in S_1] + \Pr_{(x,y)\sim\cD}[p(x) \in (t,q_1(\kappa)]|x\in S_1]< 1-\kappa + \eta,\\
        &\Pr_{(x,y)\sim\cD}[p(x) > t|x\in S_2] =1-F_2(t) \geq 1- \kappa.
    \end{align*}
    Where in the top inequality we use (\ref{eq:lipschitz}). In the second inequality we use the fact that $t>q_2(\kappa)$. 
    The two inequalities contradict (\ref{eq:dominance}), and therefore we conclude that for all $\kappa\in(K_0,1-2\eta)$, we have that $q_1(\kappa)\geq q_2(\kappa) + \beta$.
    
    We finish the proof by using the inequality on the quantiles to prove the claim. 
    
    We start by bounding $F_2(t_2)$. The conditions of this claims are stricter than those of Claim \ref{claim:max-min-dp} with $\alpha=0$, therefor we can apply it and get $\E_{(x,y)\sim\cD}[\tau(p(x),1)|x\in S_1]\leq \E_{(x,y)\sim\cD}[\tau(p(x),2)|x\in S_2]$.  This mean that $6\eta\leq c\leq \E_{(x,y)\sim\cD}[\tau(p(x),2)|x\in S_2]$, which mean that  
    \[1-F_2(t_2) = \E_{(x,y)\sim\cD}[p(x)>t_2|x\in S_2] \geq  6\eta - \eta > 3\eta.\]
    Where we use the Lipchitz condition to bound $\E_{(x,y)\sim\cD}[p(x)=t_2|x\in S_2]\leq\eta$.

    We show next that $t_1\geq t_2$. Assume towards a contradiction that $t_1<t_2$, and let $t=0.5(t_1+t_2)$ then
    \begin{align*}
            \E_{(x,y)\sim\cD}[\tau(p(x),1)|x\in S_1]=
            &\gamma_1\Pr_{(x,y)\sim\cD}[p(x)=t_1|x\in S_1]\\&+ \Pr_{(x,y)\sim\cD}[p(x)\in(t_1,t]|x\in S_1]+\Pr_{(x,y)\sim\cD}[p(x)>t|x\in S_1]
            \\\E_{(x,y)\sim\cD}[\tau(p(x),1)|x\in S_2]=&-
            (1-\gamma_2)\Pr_{(x,y)\sim\cD}[p(x)=t_2|x\in S_1]\\&- \Pr_{(x,y)\sim\cD}[p(x)\in [t,t_2)|x\in S_2]+\Pr_{(x,y)\sim\cD}[p(x)>t|x\in S_2].
        \end{align*}
    We note that $t<t_2\leq t_M$, since we have that $F_2(t_2)<1-3\eta$ and $F_2(t_M)\geq 1-\eta$. From the assumptions, we also have that $t > t_1 \geq t_0$, else we violate the capacity requirement. Therefore, $t$ is strictly in the area of a gap tail dominance, and $\Pr_{(x,y)\sim\cD}[p(x)>t|x\in S_1] > \Pr_{(x,y)\sim\cD}[p(x)>t|x\in S_2] + (t_M-t)\eta$, contradicting the equality requirement from Claim \ref{claim:max-min-dp}.
  
    Therefore, in the rest of the proof, we have that $t_1\geq t_2$. We divide into cases:
    \begin{enumerate}
        \item If $F_1(t_1)\leq F_2(t_2)$, we denote $F_2(t_2)=F_1(t_1)+\delta$.
        \begin{align*}
            \E_{(x,y)\sim\cD}[\tau(p(x),1)p(x)|x\in S_1]=&\Pr_{(x,y)\sim\cD}[p(x)=t_1|x\in S_1]\gamma_1\cdot t_1+\int_{F_1(t_1)}^{F_2(t_2)} q_1(\kappa)d\kappa\\&+\int_{F_2(t_2)}^{1-2\eta} q_1(\kappa)d\kappa+\int_{1-2\eta}^{1} q_1(\kappa)d\kappa\\\geq&
            \Pr_{(x,y)\sim\cD}[p(x)=t_1|x\in S_1]\gamma_1\cdot t_1 + \delta t_1 +\int_{F_2(t_2)}^{1-2\eta} q_2(\kappa)d\kappa\\&+ \beta(1-2\eta-F_2(t_2)) + \int_{1-2\eta}^{1} q_2(\kappa)d\kappa. \\
            \E_{(x,y)\sim\cD}[\tau(p(x),2)p(x)|x\in S_2]=&\Pr_{(x,y)\sim\cD}[p(x)=t_2|x\in S_2]\gamma_2\cdot t_2+\int_{F_2(t_2)}^{1-2\eta} q_2(\kappa)d\kappa+\int_{1-2\eta}^{1} q_2(\kappa)d\kappa
        \end{align*}
        Requiring (\ref{eq:equal-p-exp-2}) implies that
        \begin{align*}
            \Pr_{(x,y)\sim\cD}[p(x)=t_2|x\in S_2]\gamma_2 \geq \Pr_{(x,y)\sim\cD}[p(x)=t_1|x\in S_1]\gamma_1 + \delta + \frac{\beta(1-2\eta-F_2(t_2))}{t_2},
        \end{align*}
        finishing the proof for this case.

        \item If $F_1(t_1)> F_2(t_2), t_1\geq t_2$, we denote $F_1(t_1)=F_2(t_2)+\Delta$.
        \begin{align*}
            \E_{(x,y)\sim\cD}[\tau(p(x),1)p(x)|x\in S_1]=&\Pr_{(x,y)\sim\cD}[p(x)=t_1|x\in S_1]\gamma_1 t_1+\int_{F_1(t_1)}^{1-2\eta} q_1(\kappa)d\kappa+\int_{1-2\eta}^{1} q_1(\kappa)d\kappa\\\geq&
            \Pr_{(x,y)\sim\cD}[p(x)=t_1|x\in S_1]\gamma_1\cdot t_1 +\int_{F_1(t_1)}^{1-2\eta} q_2(\kappa)d\kappa\\&+ \beta(1-2\eta-F_1(t_1)) + \int_{1-2\eta}^{1} q_2(\kappa)d\kappa. \\
            \E_{(x,y)\sim\cD}[\tau(p(x),2)p(x)|x\in S_2]=&\Pr_{(x,y)\sim\cD}[p(x)=t_2|x\in S_2]\gamma_2  t_2+\int_{F_2(t_2)}^{F_1(t_1)} q_2(\kappa)d\kappa\\&+\int_{F_1(t_1)}^{1-2\eta} q_2(\kappa)d\kappa+\int_{1-2\eta}^{1} q_2(\kappa)d\kappa \\\leq&
            \Pr_{(x,y)\sim\cD}[p(x)=t_2|x\in S_2]\gamma_2\cdot t_2+t_1\Delta\\&+\int_{F_2(t_2)}^{1-2\eta}q_2(\kappa)d\kappa +\int_{1-2\eta}^{1} q_2(\kappa)d\kappa
        \end{align*}
        Requiring (\ref{eq:equal-p-exp-2}) implies that 
        \begin{align*}
            \Pr_{(x,y)\sim\cD}[p(x)=t_2|x\in S_2]\gamma_2 
            \geq \Pr_{(x,y)\sim\cD}[p(x)=t_1|x\in S_1]\gamma_1 - \Delta + \frac{\beta(1-2\eta-F_2(t_2))}{t_2}.
        \end{align*}
        This finishes the proof for this case, as $F_2(t_2)\leq 1-3\eta$.
    \end{enumerate}
\end{proof}

In the preliminaries we said that we also prove the theorem for the case of enforcing the fairness definition directly on the distribution, this is done in the next claim.
\begin{theorem}\label{thm:eo-diff-err}
     Let $\cD$ be a distribution over $\cX\times\{0,1\}$ where $
    \cX$ is divided into two groups $S_1,S_2$.
    Let $p:\cX\rightarrow[0,1]$ be a predictor with calibration error at most $\alpha$ on $S_1,S_2$, such that group $S_1$ has $t_0$-tail dominance over $S_2$. Let $c\in[0,1]$ be such that such that $\E_{(x,y)\sim\cD}[t\geq t_0\one(x\in S_i)]\geq c$.
    
    If $\tau(p(x),g(x))$ satisfies max-min fairness up to an error at most $\epsilon$, 
    \begin{align}
        \abs{\E_{(x,y)\sim\cD}[\tau(p(x),1)y|x\in S_1]- \E_{(x,y)\sim\cD}[\tau(p(x),2)y|x\in S_2]}\leq\epsilon \label{eq:condition-gap}
    \end{align}
    then 
    \begin{align}\label{eq:stricter-than-dp-err}
        \E_{(x,y)\sim\cD}[\tau(p(x),1)|x\in S_1]\leq \E_{(x,y)\sim\cD}[\tau(p(x),2)|x\in S_2] + \frac{\alpha+\epsilon}{t_{\min}}.
    \end{align}
    where $t_{\min} = \min\{t_1,t_2\}$, the thresholds of $\tau$ on groups $S_1,S_2$ respectively.
\end{theorem}

\begin{proof}
    We prove the theorem by using Claim \ref{claim:max-min-dp}. We start by showing that we satisfy the conditions of the claim.
    Since $p$ is $\alpha$-calibrated, from Claim \ref{claim:calibration} and (\ref{eq:condition-gap}) we have that
    \begin{align*}
        \E_{(x,y)\sim\cD}[\tau(p(x),1)p(x)|x\in S_1]\leq  \E_{(x,y)\sim\cD}[\tau(p(x),2)p(x)|x\in S_2] + \alpha+\epsilon,
    \end{align*}
    From the theorem statement, we have that $p$ on $x\in S_1$ is $t_0$-tail dominance over $p(x)$ on $x\in S_2$. From the capacity requirement, both thresholds $t_1,t_2\geq t_0$. Therefore, we satisfy the conditions of Claim \ref{claim:max-min-dp} with parameter $\alpha'=\alpha+\epsilon$, and finish the proof.
\end{proof}

\begin{proof}[Proof of \ref{claim:additive-price-of}]
     Let $\tau$ be the randomized group threshold function satisfying max-min fairness and $\tau_f$ the randomized group threshold function maximizing the true positive count. Then $\tau_f$ has a single threshold function for the two groups, denoted by $t_f$. Since $S_1$ is tail-dominant over $S_2$ and from our capacity $t_f\in [t_0,t_M]$, we have that $F_1(t_f)\leq F_2(t_f)-\eta$.
    From Theorem \ref{thm:eo-diff-gap}, $F_1(t_1)\geq F_2(t_2)+\eta\beta/t_2$.
    
    We know that $\tau_f$ and $\tau$ have the same capacity $c$. This implies that:
    \begin{align*}
        c=&\Pr_{x\sim\cD}[x\in S_1](1-F_1(t_1))+\Pr_{x\sim\cD}[x\in S_2](1-F_2(t_2))\\
        F_1(t_1)\geq& 1-c+\frac{\eta\beta}{t_2}\Pr_{x\sim\cD}[x\in S_2]\\
        F_2(t_2)\leq& 1-c- \frac{\eta\beta}{t_2}\Pr_{x\sim\cD}[x\in S_1]\\
        c=&\Pr_{x\sim\cD}[x\in S_1](1-F_1(t_f))+\Pr_{x\sim\cD}[x\in S_2](1-F_2(t_f))\\
        F_1(t_f)\leq&1-c-\eta\Pr_{x\sim\cD}[x\in S_2]\\
        F_2(t_f)\geq& 1-c+\eta \Pr_{x\sim\cD}[x\in S_1].
    \end{align*}
    Since our groups are equal-sized, we have that $F_1(t_1)-F_1(t_f)\geq \left(\frac{\eta\beta}{t_2}+\eta\right)\frac{1}{2}$, and
    $F_1(t_1)-F_1(t_f)=F_2(t_f)-F_2(t_2)$.
     
    Since the distribution of $p(x),x\sim\cD$ is continuous for $x\in S_1,S_2$, we have that
    \begin{align*}
        \Pr_{(x,y)\sim\cD}[p(x)\in(t_f,t_1),x\in S_1]=\Pr_{(x,y)\sim\cD}[p(x)\in(t_2,t_f),x\in S_2],
    \end{align*}
 
    Using Claim \ref{claim:quantile-exp},
    \begin{align*}
        \E_{(x,y)\sim\cD}[\tau_f(p(x),g(x))p(x)]=&\Pr_{x\sim\cD}[x\in S_1]\int_{F_1(t_f)}^1q_1(\kappa)d\kappa +\Pr_{x\sim\cD}[x\in S_2]\int_{F_2(t_f)}^1q_2(\kappa)d\kappa\\=&
        \Pr_{x\sim\cD}[x\in S_1]\int_{F_1(t_f)}^{F_1(t_1)}q_1(\kappa)d\kappa +\Pr_{x\sim\cD}[x\in S_1]\int_{F_1(t_1)}^1q_1(\kappa)d\kappa\\&+
        \Pr_{x\sim\cD}[x\in S_2]\int_{F_2(t_2)}^{1}q_2(\kappa)d\kappa -\Pr_{x\sim\cD}[x\in S_2]\int_{F_2(t_2)}^{F_2(t_f)}q_2(\kappa)d\kappa\\=&
        \E_{(x,y)\sim\cD}[\tau(p(x),g(x))p(x)]+\Pr_{x\sim\cD}[x\in S_1]\int_{F_1(t_f)}^{F_1(t_1)}q_1(\kappa)d\kappa\\&-\Pr_{x\sim\cD}[x\in S_2]\int_{F_2(t_2)}^{F_2(t_f)}q_2(\kappa)d\kappa
    \end{align*}
    We bound the integrals. 
    From the Lipchitz condition, at most $\eta/4$ of the weight can be in a range $[r,r+\beta]$. Therefore, we can lower bound the probability as follows:
     \begin{align*}
        \E_{x\sim\cD}[\one(p(x)\in (t_f,t_1))p(x)|x\in S_1] \geq& \frac{\eta}{4} t_f + \left(\Pr_{x\sim\cD}[p(x)\in (t_f,t_1)|x\in S_1]-\frac{\eta}{4}\right)\left(t_f+\beta \right)\\\geq&
        \Pr_{x\sim\cD}[p(x)\in (t_f,t_1)|x\in S_1](t_f+\beta)-\frac{\eta\beta}{4}.
    \end{align*}
    For the second integral, 
    \begin{align*}
        \E_{x\sim\cD}[\one(p(x)\in (t_2,t_f))p(x)|x\in S_2] \leq& t_f\Pr_{x\sim\cD}[p(x)\in (t_2,t_f)|x\in S_2]
    \end{align*}

    Combining it together we have that
    \begin{align*}
        \E_{(x,y)\sim\cD}[\tau_f(p(x),g(x))p(x)]\geq& \E_{(x,y)\sim\cD}[\tau(p(x),g(x))p(x)] + \Pr_{x\sim\cD}[p(x)\in (t_f,t_1),x\in S_1](t_f+\beta)\\&-\Pr_{x\sim\cD}[x\in S_1]\frac{\beta\eta}{4}-t_f\Pr_{x\sim\cD}[p(x)\in (t_2,t_f),x\in S_2]
        \\\geq&\E_{(x,y)\sim\cD}[\tau(p(x),g(x))p(x)]\\&+\Pr_{x\sim\cD}[x\in S_1]\beta\left(\Pr_{x\sim\cD}[p(x)\in (t_f,t_1)|x\in S_1]-\frac{\eta}{4}\right)\\\geq&
        \E_{(x,y)\sim\cD}[\tau(p(x),g(x))p(x)] +\frac{\eta\beta}{8} + \frac{\eta\beta^2}{4t_2}.
    \end{align*}
    Where we use the fact that the groups are equal sized.
    
\end{proof}

\begin{proof}[Proof of Theorem \ref{thm:prop-bound}]
    Let $p$ as in the theorem statement. Denote by $t_1,t_2$ the thresholds of $\tau$ for groups $S_1,S_2$ respectively.

    By definition, the post-processing $\tau$ that is a randomized threshold function satisfies propositional fairness on the simulated distribution if
    \begin{align*}
        \tau = \arg\max_{\tau''}\sum_{i\in[2]}\Pr_{(x,y)\sim\cD}[x\in S_i]\log\E_{(x,y)\sim\cD}[\tau''(p(x),i)p(x) | x\in S_i].
    \end{align*}
     Denote by $c_1\leq c$ the fraction of the capacity that $S_1$ receives, $c_1=\E_{(x,y)\sim\cD}[\tau(p(x),1)\one(x\in S_1)]$. Since we have full utilization and only two groups, $\tau$ is fully characterized by $c_1$. 

     Let $c_{\max} = \min\{c,\Pr_{x\sim\cD}[x\in S_1]\}$ and $c_{\min} = \max\{0, c-\Pr_{x\sim\cD}[x\in S_2]\}$.
     For every $k\in[c_{\min},c_{\max}]$ we can define $\tau_k$ as the randomized group threshold function such that $k=\E_{(x,y)\sim\cD}[\tau_k(p(x),1)\one(x\in S_1)]$ and $\E_{(x,y)\sim\cD}[\tau_k(p(x),2)\one(x\in S_2)]=c-k$. Then the capacity of $\tau_k$ is exactly $c$. From Claim \ref{claim:existence-threshold}, we can always find such threshold function $\tau_k$.
     
     Using this notation, $c_1$ is solution to the following optimization problem:
     \begin{align}\label{eq:max-c1}
         c_1 = \arg\max_{k}\sum_{i\in[2]}\Pr_{(x,y)\sim\cD}[x\in S_i]\log\E_{(x,y)\sim\cD}[\tau_k(p(x),i)p(x) | x\in S_i]
     \end{align}
    We have how an optimization problem with a single variable. For every $k$, denote by $H_i(k)=\E_{(x,y)\sim\cD}[\tau_k(p(x),i)p(x) | x\in S_i]$. Then 
    $\text{Prop}(k)=\sum_{i\in[2]}\Pr_{(x,y)\sim\cD}[x\in S_i]\log H_i(k)$. 

    If $c_{\min}=0,c_{\max}=c$, then we know that the function receives its maximal value in the interim of the part. Else, we need to check the edges $c_{\min},c_{\max}$, which we do at the end of the proof.
    Therefore, we can characterize the optimal value by calculating the derivative of $\text{Prop}(k)$ and comparing it to zero.
    \begin{align}\label{eq:prop-derivative}
        \frac{d\text{Prop}(k)}{dk}=\sum_{i\in[2]}\Pr_{(x,y)\sim\cD}[x\in S_i]\frac{1}{H_i(k)}\frac{d H_i(k)}{dk}.
    \end{align}
    In order to use (\ref{eq:prop-derivative}), we need to calculate the derivative of $H_i(k)$. We denote by $Q_i$ the distribution of $p(x),x\in S_i$. Then let $q_i$ be the quantile function of $Q_i$ and $F_i$ be the CDF of this distribution. Let $t_i(k)$ be the threshold value of $\tau_k$ for group $i$. 
    Then in this notation, using Claim \ref{claim:quantile-exp}:
    \begin{align*}
        H_i(k)=&\E_{(x,y)\sim\cD}[\tau_k(p(x),i)p(x) | x\in S_i]=\Pr_{x\sim\cD}[p(x)=t_i(k)|x\in S_i]t_i(k)\gamma_i(k)+\int_{F_i(t_i(k)))}^1 q_i(\kappa)d\kappa
    \end{align*}
    Where we have that $t_i(k)=q_i(1-w_i)$ where $w_1=\frac{k}{\Pr_{x\sim\cD}[x\in S_1]}$, and $w_2=\frac{c-k}{\Pr_{x\sim\cD}[x\in S_2]}$ 
    \begin{align*}
    \gamma_i(k) = \begin{cases}
        0   \quad &\Pr_{u\sim Q_i}[u=t_i]=0,\\
        \frac{w_i-1+F_i(t_i))}{\Pr_{u\sim Q_i}[u=t_i]} \quad &\text{o.w.}
    \end{cases}
\end{align*}
    To calculate the derivative we divide into cases, by the cases in $\gamma_i(k)$. If $\gamma_i(k)>0$, then $\frac{d t_i(k)}{dk}=0$, as the quantile function is constant at $k$. In this case, 
    \begin{align*}
        \frac{d H_i(k)}{dk}=&\Pr_{x\sim\cD}[p(x)=t_i(k)|x\in S_i]t_i(k)\frac{d\gamma_i(k)}{dk}\\=&\frac{\Pr_{x\sim\cD}[p(x)=t_i(k)|x\in S_i]t_i(k)}{\Pr_{x\sim\cD}[p(x)=t_i(k)|x\in S_i]\Pr_{x\sim\cD}[x\in S_i]}(-1)^{i-1}\\=&\frac{t_i(k)}{\Pr_{x\sim\cD}[x\in S_i]}(-1)^{i-1}.
    \end{align*}
    
    In the case that $\gamma_i(k)=0$, the quantile function is not constant at $k$ but the first term is zero. In addition, it mean that $F_i(t_i(k))=F_i(q_i(1-w_i))=1-w_i$.
    In this case we use the Leibnitz integral rule:
    \begin{align*}
        \frac{d H_i(k)}{dk}=-q_i(1-w_i(k))\frac{d(1-w_i(k))}{dk}=-t_i(k)\left(-\frac{1}{\Pr_{x\sim\cD}[x\in S_i]}\right)(-1)^{i-1}.
    \end{align*}
    In both cases, the derivative is $\frac{t_i(k)}{\Pr_{x\sim\cD}[x\in S_i]}(-1)^{i-1}$. We can see it intuitively by noticing that if we increase the total capacity given to $S_1$ by $\epsilon$, we increase the conditional capacity by $\epsilon/\Pr_{x\sim\cD}[x\in S_i]$, and the expected true positive count by this time the current threshold, as this is the probability of $y=1$ of the individuals we change their allocation. The minus sign for $i=2$ is because the dependency of $w_i$ in $k$ is negative for $S_2$.

    Therefore, the optimal allocation $c_1$, which has derivative $0$ satisfies:
    \begin{align*}
        \Pr_{(x,y)\sim\cD}[x\in S_1]\frac{1}{H_1(c_1)}\frac{t_1(c_1)}{\Pr_{(x,y)\sim\cD}[x\in S_1]} + \Pr_{(x,y)\sim\cD}[x\in S_2]\frac{1}{H_2(c-c_1)}\left( - \frac{t_2(c-c_1)}{\Pr_{(x,y)\sim\cD}[x\in S_2]}\right)=0\\
    \end{align*}
    That is,
    \begin{align}\label{eq:cond-zero-der}
        \frac{t_1(c_1)}{H_1(c_1)}=\frac{t_2(c-c_1)}{H_2(c-c_1)}.
    \end{align}

    We use (\ref{eq:cond-zero-der}) to bound the gap between the total true positive count of the proportional fairness comparing to the true positive count of the post-processing without any fairness requirements.

    The post-processing $\tau'$ is the post-processing that maximizes the true positive count, and denote by $c_i' = \E_{(x,y)\sim\cD}[\tau'(x,i)\one(x\in S_i)]$. Assume without loss of generality that $c_1' \geq c_1$. Since $\tau'$ is the post-processing that maximizes the true positive count, it means that $t_1(c_1) \geq t_2(c-c_1)$.
    
    For each individual $x$ that $\tau'$ allocates more resources than $\tau$, that is for $x$ such that $\tau(p(x),g(x))<\tau'(p(x),g(x))$, we have that $p(x)\leq t_1(c_1)$ (as $\tau$ already allocated resources to all individuals with $p(x)>t_1(c_1)$ in both groups). This means that 
    \begin{align}\label{eq:max-utils-bound}
    M=\E_{(x,y)\sim\cD}[\tau'(p(x),g(x))p(x)]\leq \E_{(x,y)\sim\cD}[\tau(p(x),g(x))p(x)]+(c_1'-c_1)t_1(c_1).  
    \end{align}
    
    In order to bound the gap we lower bound the true positive count of $\tau$ in terms of $t_1(c_1)$.
    For $S_2$, we have that
    $H_2(c-c_1)\geq t_2(c-c_1)\Pr_{(x,y)\sim\cD}[p(x)>t_2(c-c_1)|x\in S_2] \geq t_2(c-c_1)\frac{(c-c_1)}{\Pr_{x\sim\cD}[x\in S_2]}$, as $\tau$ allocate capacity $c-c_1$ to group $S_2$ with threshold $t_2(c-c_1)$.

    Combining this with (\ref{eq:cond-zero-der}):
    \begin{align*}
        H_1(c_1) = \frac{H_2(c-c_1)t_1(c_1)}{t_2(c-c_1)}\geq \frac{c-c_1}{\Pr_{x\sim\cD}[x\in S_2]}t_1(c_1).
    \end{align*}
    Using this bound:
    \begin{align*}
        \E_{(x,y)\sim\cD}[\tau(p(x),g(x))p(x)]\geq \Pr_{x\sim\cD}[x\in S_1]H_1(c_1) \geq (c-c_1)t_1(c_1)\frac{\Pr_{x\sim\cD}[x\in S_1]}{\Pr_{x\sim\cD}[x\in S_2]}.
    \end{align*}
    Since $S_1,S_2$ are equal size, we have that $\E_{(x,y)\sim\cD}[\tau'(p(x),g(x))p(x)]\leq 2\E_{(x,y)\sim\cD}[\tau(p(x),g(x))p(x)]$.
    Using Claim \ref{claim:calibration}, the difference between $\E_{(x,y)\sim\cD}[\tau(p(x),g(x))p(x)]$ and $\E_{(x,y)\sim\cD}[\tau(p(x),g(x))y]$ is at most $\alpha$, which finishes the proof for the standard case.

    We are left with handling the case that (\ref{eq:cond-zero-der}) does not hold for any $c_1$ in our range, i.e. that our function receives its maximal value at one of the edges. This case is only possible if the capacity $c$ is larger than one the groups size (else $c_{\min}=0,c_{\max}=c$ and the optimization in the edges approaches $-\inf$). This is not the standard setting considered in this work, but we prove it for completeness. 
    Since the setting is symmetric with respect to the groups, without loss of generality $\tau'$, the post-processing maximizing the true positive count has $c_1' \geq c_1$. As we are in the case that the maximum of $\text{Prop}$ is at one of the edges, this edge must be $c_{\min}$, i.e. $\E_{(x,y)\sim\cD}[\tau(p(x),2)|x\in S_1]=1$.  
    If $\text{Prop}$ receives the maximal value at $c_{\min}$, we have that its derivative is always negative, and instead of (\ref{eq:cond-zero-der}) we have $\frac{t_1(c_1)}{H_1(c_1)}\leq\frac{t_2(c-c_1)}{H_2(c-c_1)}$. We note that the proof only uses inequality, and therefore we can use the exactly same proof for this case as well.
    
\end{proof}

\section{Proofs of Section \ref{sec:new-def}}\label{app:new-def}
\begin{proof}[Proof of Claim \ref{claim:new-def}]
For every $i\in[m]$, let $f_i=\arg\max_{f'\in \cF_c}\E_{(x,y)\sim\cD}[f'(x)y]$. Let $f$ be the function defined by $f^*(x)=f_{g(x)}(x)$. Since $\cF$ is closed under group choice post-processing, we have that $f^*\in\cF$.
We note that since each $f_i$ can allocate all of the resources only to $S_i$, the capacity of $f^*$ can be $mc$.
By definition, $\E[f^*(x)y]=\sum_{i\in[m]}\Pr_{(x,y)\sim\cD}[x\in S_i]\cA(\cF,c,S_i)$.

Since $h$ has capacity $c$ over all groups and $f^*$ picks for every group the function in $\cF$ with capacity $c$ that maximizes the true positive count, we have that
\begin{align*}
    \E_{(x,y)\sim\cD}[f^*(x)y]\geq \E_{(x,y)\sim\cD}[h(x)y].
\end{align*}

We prove next that $\E_{(x,y)\sim\cD}[f(x)y]\geq \E_{(x,y)\sim\cD}[f^*(x)y]/m$, which finished the proof.

We start by explicitly writing the true positive count of $f$, the function satisfying achievable equal opportunity:
\begin{align*}
    \E_{(x,y)\sim\cD}[f(x)y]=\sum_j \Pr_{(x,y)\sim\cD}[x\in S_j]\E_{(x,y)\sim\cD}[f(x)y|x\in S_i].
 \end{align*}
Let $i\in[m]$ be an arbitrary index, then from the achievable equal opportunity requirement,
\begin{align*}
    \E_{(x,y)\sim\cD}[f(x)y]=&\sum_j \Pr_{(x,y)\sim\cD}[x\in S_j]\frac{\E_{(x,y)\sim\cD}[f(x)y|x\in S_i]\cA(\cF,c,S_j)}{\cA(\cF,c,S_i)}\\=&\frac{\E_{(x,y)\sim\cD}[f(x)y|x\in S_i]}{\cA(\cF,c,S_i)}\sum_j \Pr_{(x,y)\sim\cD}[x\in S_j]\cA(\cF,c,S_j)\\=&
    \frac{\E_{(x,y)\sim\cD}[f(x)y|x\in S_i]}{\cA(p,c,S_i)}\E_{(x,y)\sim\cD}[f^*(x)y].
 \end{align*}
 The above equality holds for every $i$, and specifically let $i^*\in[m]$ be such that $\E_{(x,y)\sim\cD}[f(x)|x\in S_{i^*}]\geq c$ (by averaging, we must have at least one such $i^*$ as the capacity of $f$ is $c$).
 For $i^*$, 
 \begin{align*}
     \E_{(x,y)\sim\cD}[f(x)y|x\in S_{i^*}]\geq \cA(\cF,c,S_{i^*})/m.
 \end{align*}
This is because $\cF$ is closed under group choice post-processing, and therefore we could have set $f(x) = 1/m f_{i^*}(x)$ for all $x\in S_{i^*}$ and $0$ for all other groups. It is possible, of course, that there is a better choice for $f$ on $S_{i^*}$, and in this case the inequality will be strict.

Combining everything together, we get that 
\begin{align*}
    \E_{(x,y)\sim\cD}[f(x)y] \geq \frac{1}{m}\E_{(x,y)\sim\cD}[f^*(x)y] \geq \E_{(x,y)\sim\cD}[h(x)y].
\end{align*}
\end{proof}

We show next that the above bound is tight, and we can find $m$ groups such that requiring achievable equal opportunity reduces the expected true positive count by a factor approaching $m$.
\begin{example}
    Let $\cD$ be a distribution over $\cX\times\{0,1\}$, and assume that $\cX$ is divided into $m$ equal-sized groups.
    Let $\delta< 1/m$ be a small constant, and let $k <  1/(m\delta c)$.
    Let $p:\cX\rightarrow[0,1]$, such that on all $x\in S_1$, $p(x)=k\delta$ and for all $i>1$ we have $p(x)=\delta$ for all $x\in S_i$. This mean that for all $c\leq 1/m$, $\cA(p,c,S_1)=  k\delta mc$ and for all $i>1$,$\cA(p,c,S_i)=\delta mc$.

    The function $h:\cX\rightarrow[0,1]$ that has the maximal true positive count has $h(x)=0$ to all $x\notin S_1$ and $h(x)=c m$ for all $x\in S_1$.
    It satisfies:
    \begin{align*}
        \E_{(x,y)\sim\cD}[h(x)y]=\Pr_{(x,y)\sim\cD}[x\in S_1]\E_{(x,y)\sim\cD}[h(x)y|x\in S_1]=\frac{1}{m}ck\delta m = ck\delta.
    \end{align*}

    Let $f$ be the function maximizing the true positive count while satisfying the achievable equal opportunity constraint. 
    For such $f$, by definition,
    \begin{align*}
        \frac{\E_{(x,y)\sim\cD}[f(x)y|x\in S_1]}{\delta k m c} = \frac{\E_{(x,y)\sim\cD}[f(x)y|x\in S_i]}{\delta m c}.
    \end{align*}
    That is, we must have $\E_{(x,y)\sim\cD}[f(x)y|x\in S_1]=k \E_{(x,y)\sim\cD}[f(x)y|x\in S_i]$.
    For our function $p$, this implies that $f(x)=c$ for all $x$, as then we have $\E_{(x,y)\sim\cD}[f(x)]=c$, $\E_{(x,y)\sim\cD}[f(x)y|x\in S_1]=k\delta c$ and for all $i>1$, $\E_{(x,y)\sim\cD}[f(x)y|x\in S_i]=\delta c$ as required by the equal opportunity condition.

    The true positive count of $f$ is given by:
    \begin{align*}
        \E_{(x,y)\sim\cD}[f(x)y]=\frac{1}{m}k\delta c+\frac{m-1}{m}\delta c\leq\delta c\left(\frac{k}{m} + 1 \right).
    \end{align*}
    When $k$ is large, we get that $\E_{(x,y)\sim\cD}[f(x)y]\approx \frac{1}{m}\E_{(x,y)\sim\cD}[h(x)y]$, as required.
\end{example}
\section{Simulation}\label{app:simulations}
We start by describing the simulations appearing in the main part of the paper.
\paragraph{Figure \ref{fig:intro-single-threshold}:} In this figure we have two histograms, representing $p(x)$ on $x\in S_1$ and $p(x)$ on $x\in S_2$. In this setting, we generated the distributions of $p(x)$ on each group as a Gaussian random variable, both with mean $0.05$ representing hard to predict and relatively rare outcome. For group $S_1$ the variance is chosen to be $0.02$ and for group $S_2$ the variance is chosen to be $0.01$. The values of $p(x)$ are then clipped to $[0,1]$ as it is supposed to be the output of a risk predictor. We remark that this implies that the actual mean is not $0.05$, and it differs between the groups. The cell $p(x)=0$ is omitted from the histogram, as it is huge comparing to the rest of the bins due to the clipping and does not provide any useful information.

The line drawn is a single threshold, which is chosen such that $1\%$ of the population are above the threshold (when $S_1,S_2$ have equal size). In this histogram, over $80\%$ of the individuals above the cutoff belong to group $S_1$.

\paragraph{Figure \ref{fig:price-large-gap}:} This figure demonstrate the fairness requirements where we have a very large gap between the two groups. The distribution of $p(x)$ on $x\in S_1$ has variance $0.00005$ and on $x\in S_2$ has variance $0.000005$ (i.e. on $S_1$ it is $10$ times larger). Both distributions have a mean value of $0.00235$ after clipping (i.e. about 23 positive $y=1$ out of $1000$ individuals). We assume that both groups $S_1,S_2$ are of the same size.

Given a capacity of $1\%$, the allocation maximizing the true positive count allocate all of its resources to individuals from $S_1$. The true positive count is $0.00016$, i.e. in this allocation we find about $1.6$ out of the $23$ positives for every $1000$ people.

Since we have equal base rates, max-min fairness and equal opportunity are the same. The allocation satisfying max-min fairness allocates slightly more than $25\%$ of its resources to $S_1$, and the large majority to $S_2$. The true positive count is $0.00011$, i.e. we find about $1.1$ out of the $23$ trues positives for each $1000$ individuals.

When requiring proportional fairness, the allocation allocates about half of the resources to each of the groups. Therefore, the true positive count is $0.00013$, i.e. we find about $1.3$ out of the $23$ trues positives for each $1000$ individuals. We can see that in this case, the proportional fairness is a ``middle ground'' requirement, between max-min fairness and not requiring any fairness requirements.

\paragraph{Figure \ref{fig:price-large-capacity}:} In this figure we demonstrate the behaviors of the different fairness requirements where we have a larger capacity, of $15\%$ of the population. In addition, instead of a larger variance for $S_1$ we chose a smaller variance with a different mean value, to demonstrate another type of tail dominance. The distribution of $p(x)$ on $S_1$ has mean $0.005$ (before clipping) and variance $0.000005$, while the distribution of $p(x)$ on $S_2$ has mean $0.002$ and variance $0.000008$.
The mean after clipping is $0.0036$, i.e. corresponds to $36$ individuals out of $1000$ with positive outcome $y=1$.

Given such higher capacity, the allocating optimizing the total true positive count without fairness requirements has true positive count of $0.00117$, i.e. finds nearly a third of the $y=1$. It also allocates less than $20\%$ of the resources to individuals from $S_2$.

The allocation satisfying max-min fairness allocate $43\%$ of the capacity to $S_1$, and the rest to $S_2$. The total true positive count is $0.00108$, indicating a smaller gap in the setting of large capacity.

The allocation satisfying proportional fairness allocates about $48\%$ of the capacity to $S_2$, and has true positive rate of $0.001124$. While the gaps are smaller, we can see that the proportional fairness is in between max-min fairness and a completely utilitarian approach.

\end{document}